\documentclass{article}
\usepackage{amsmath,amssymb,amsthm}
\usepackage{graphicx}
\usepackage{hyperref}
\usepackage{tikz}
\usepackage{caption}
\usepackage{subcaption}
\usetikzlibrary{math}
\usepackage{float}
\usepackage{fancyhdr}

\newcommand{\ket}[1]{|#1\rangle}

\newcommand{\F}{\mathbb{F}_2}
\newcommand{\Z}{\mathbb{Z}}
\newcommand{\N}{\mathbb{N}}

\newcommand{\onebytwo}[2]{\begin{pmatrix} #1 & #2\end{pmatrix}}
\newcommand{\twobytwo}[4]{\begin{pmatrix} #1 & #2 \\ #3 & #4\end{pmatrix}}
\newcommand{\threebytwo}[6]{\begin{pmatrix} #1 & #2 \\ #3 & #4 \\ #5 & #6\end{pmatrix}}
\newcommand{\twobyone}[2]{\begin{pmatrix} #1 \\ #2\end{pmatrix}}
\newcommand{\onebythree}[3]{\begin{pmatrix} #1 & #2 & #3\end{pmatrix}}
\newcommand{\threebyone}[3]{\begin{pmatrix} #1 \\ #2 \\ #3\end{pmatrix}}
\newcommand{\twobythree}[6]{\begin{pmatrix} #1 & #2 & #3 \\ #4 & #5 & #6 \end{pmatrix}}
\newcommand{\threebythree}[9]{\begin{pmatrix} #1 & #2 & #3 \\ #4 & #5 & #6 \\ #7 & #8 & #9 \end{pmatrix}}
\newcommand{\onebyfour}[4]{\begin{pmatrix} #1 & #2 & #3 & #4 \end{pmatrix}}
\newcommand{\twobyfour}[8]{\begin{pmatrix} #1 & #2 & #3 & #4 \\ #5 & #6 & #7 & #8 \end{pmatrix}}

\newcommand{\PauliX}{\sigma_x}
\newcommand{\PauliY}{\sigma_y}
\newcommand{\PauliZ}{\sigma_z}
\newcommand{\SymplecticProductMatrix}{\Omega}
\newcommand{\PauliGroup}[1]{\mathcal{P}_{#1}}
\newcommand{\PauliGroupHomomorphism}{\Phi}
\newcommand{\CliffordGroup}[1]{\mathcal{CL}_{#1}}
\newcommand{\Normalizer}[2]{\mathcal{N}_{#2}\left(#1\right)}
\newcommand{\UnitaryGroup}[1]{\mathcal{U}_{#1}}
\newcommand{\SymplecticGroup}[1]{\mathcal{SP}_{2 #1}}
\newcommand{\CliffordGroupHomomorphism}{\Psi}

\newcommand{\GP}[1]{\Pi\left(#1\right)}

\newcommand{\ones}[1]{\vec{1}_{#1}}

\newtheorem{theorem}{Theorem}
\newtheorem{definition}{Definition}

\newtheorem{corollary}{Corollary}

\newcommand{\AND}[2]{#1 \text{ AND } #2}

\newcommand{\NOT}[1]{\text{NOT}(#1)}

\newcommand{\MIN}[1]{\mathcal{MIN}\left(#1\right)}
\newcommand{\MAX}[1]{\mathcal{MAX}\left(#1\right)}
\newcommand{\GUP}[1]{\left\langle #1 \right\rangle_{\uparrow}}
\newcommand{\GDO}[1]{\left\langle #1 \right\rangle_{\downarrow}}

\newcommand{\GRM}[1]{\mathcal{GRM}\left(#1\right)}
\newcommand{\GRMT}[1]{\mathcal{GRM}'\left(#1\right)}

\begin{document}
\title{Quantum LDPC codes from intersecting subsets}
\author{Dimiter Ostrev}
\date{German Aerospace Center (DLR), Institute of Communications and Navigation, 82234 We{\ss}ling, Germany}
\maketitle

\begin{abstract}
This paper introduces a construction of quantum CSS codes from a tuple of component CSS codes and two collections of subsets. The resulting codes have parallelizable encoding and syndrome measurement circuits and built-in redundancy in the syndrome measurements. In a certain subfamily of the general construction, the resulting codes are related to a natural generalization of classical Reed-Muller codes, and this leads to formulas for the distance of the quantum code as well as for the distance of the associated classical code that protects against errors in the syndrome. The paper gives a number of examples of codes with block size $2^m, m=3,\dots,9$, and with syndrome measurements involving 2, 4 or 8 qubits. These include codes for which the distance exceeds the syndrome measurement weight, as well as codes which provide asymmetric protection against bit flip and phase flip errors. 
\end{abstract}

%arXivAdmin remove initial copyright section
%\section{Introduction}

%\thispagestyle{fancy}
%\fancyhead{}
%\renewcommand{\headrulewidth}{0pt}
%\renewcommand{\footrulewidth}{1pt}
%\fancyfoot{}

%\addtocounter{footnote}{-1}
%\blfootnote{\centering \copyright 2024 IEEE.  Personal use of this material is permitted.  Permission from IEEE must be obtained for all other uses, in any current or future media, including reprinting/republishing this material for advertising or promotional purposes, creating new collective works, for resale or redistribution to servers or lists, or reuse of any copyrighted component of this work in other works.}

%arXivAdmin copyright format change
\section*{Copyright}

\copyright 2024 IEEE. Personal use of this material is permitted. Permission from IEEE must be obtained for all other uses, in any current or future media, including reprinting/republishing this material for advertising or promotional purposes, creating new collective works, for resale or redistribution to servers or lists, or reuse of any copyrighted component of this work in other works.

\section{Introduction}

Quantum error correction \cite{shor1995scheme,steane1996errorcorrecting} protects quantum states from the effects of noise during communication or computation. Many of the codes studied so far can be described using the stabilizer formalism \cite{gottesman1996class,calderbank1997quantum}. 

In recent years, quantum low-density parity check codes have received considerable attention \cite{mackay2004sparse,tillich2014quantum,hastings2021fiber,breuckmann2021balanced,panteleev2021asymptotically,leverrier2022quantum,gu2023efficient,leverrier2022efficient,dinur2023goodquantum}. In these codes, each syndrome measurement involves only a few qubits, and this is considered favourable for fault-tolerance. Moreover, there is hope that low-complexity message passing decoders \cite{gallager1963low,richardson2001capacity} known from classical LDPC codes can be adapted to the quantum case.

Research on quantum LDPC codes can be classified according to various criteria. This introduction considers three: the ideas motivating the code design, the axis communication-fault tolerance, and the axis finite block length-asymtotic regime. The following paragraphs give a high-level description and some references for each, then explain where the present work stands with respect to each criterion. 

\paragraph{Ideas motivating the code design}

The construction methods used for classical LDPC codes do not easily translate to their quantum counterparts, because of the constraint that stabilizer elements must commute. In recent years, much progress in quantum LDPC codes has been achieved using ideas from homological algebra, geometry and topology. Exploring in detail the many and varied results is beyond the scope of this introduction; fortunately, the recent survey \cite{breuckmann2021quantum} gives an excellent overview. Subsequent to the publication of that survey, the long term goal of constructing  asymptotically good quantum LDPC codes was achieved by lifted product codes \cite{panteleev2021asymptotically} and quantum Tanner codes \cite{leverrier2022quantum,gu2023efficient,leverrier2022efficient,dinur2023goodquantum}. 

By contrast, the present paper uses a code design that is arguably much simpler and much closer to classical coding. 

Consider the first construction of classical LDPC codes due to Gallager \cite{gallager1963low} in relationship to the earlier construction of classical product codes \cite{elias1954errorfree}. Specifically, think first of the parity check matrix of a classical product of single parity check codes. The resulting parity check matrix has layers of rows, with each layer having rows of disjoint support, whose union covers all columns. For example, the parity check matrix of the two-fold product of the [3,2,2] single parity check code with itself is $$\twobyone{\threebythree{1}{0}{0}{0}{1}{0}{0}{0}{1} \otimes \onebythree{1}{1}{1}}{\onebythree{1}{1}{1}\otimes\threebythree{1}{0}{0}{0}{1}{0}{0}{0}{1}}= \begin{pmatrix}1&1&1&0&0&0&0&0&0\\0&0&0&1&1&1&0&0&0\\0&0&0&0&0&0&1&1&1\\ \hline 1&0&0&1&0&0&1&0&0\\0&1&0&0&1&0&0&1&0\\0&0&1&0&0&1&0&0&1\end{pmatrix}$$ Now, think of Gallager's construction of classical LDPC codes \cite[Section 2.2 and Figure 2.1]{gallager1963low}. Gallager uses as a basic building block a layer of rows with disjoint support whose union covers all columns. He then builds the full parity check matrix from several such layers, each of which is a random column permutation of the first layer. Thus, Gallager's LDPC codes can be viewed as a generalization of products of single parity check codes, obtained by allowing more general column permutations for the layers.

How can Gallager's design using layers of rows be extended to quantum Calderbank-Shor-Steane codes? Since the parity check matrices for bit flips and phase flips must be orthogonal to each other, completely arbitrary column permutations of the layers are not possible. However, progress may be obtained by imposing some additional structure. Recently, Hivadi \cite{hivadi2018quantum} showed how to construct a $[16s^2,16s^2-16s+2,4]]$ quantum CSS code such that one of the parity check matrices is the classical 2-fold product of the $[4s,4s-1]$ single parity check code, and the other parity check matrix is a column permutation of the first one. Later, this was generalized in \cite{ostrev2024classicalproduct} to classical $D$-fold products for any $D$. 

The present paper presents new quantum Calderbank-Shor-Steane codes that are a natural generalization of the product code construction of \cite{ostrev2024classicalproduct}. In the context of the preceding discussion, this generalization can be viewed as introducing greater flexibility in choosing column permutations for the layers, while at the same time retaining sufficient structure to ensure that the two parity check matrices of the quantum CSS code are orthogonal. For reasons that will become clear later, this new family is called intersecting subset codes. 

Using a design with layers of rows as described above automatically leads to parity check matrices that are not full rank: for example, the sum of rows in each layer is the all-ones vector. When the column permutations of different layers are random, it seems difficult to say much more about the linear dependence of the rows. However, it turns out that the additional structure present in intersecting subset codes is enough to compute explicitly a basis for the stablizer. Moreover, it turns out that this basis consists of rows of a tensor product of invertible matrices, which points to a surprising connection to polar \cite{arikan2009channel} and Reed-Muller \cite{reed1954class,muller1954application} codes. Thus, there are at least two structured choices of a spanning set for the stabilizer of intersecting subset codes; the correspondence between them generalizes the connection between classical product and polar codes observed in \cite{coskun2023successivecancellation}.  From these two representations, a number of desirable properties of intersecting subset codes are derived. 

The first choice of a spanning set for the stabilizer corresponds to the original motivation of obtaining parity check matrices similar to Gallager's design. These can be divided into layers, with the measurements in each layer done in parallel, and the layers measured in sequence. Moreover, for suitable examples of intersecting subset codes, each measurement involves only a few qubits. 

The second choice of a spanning set for the stabilizer consists of subsets of the rows of a tensor product of invertible matrices and its inverse transpose. From this, an encoding circuit follows, which can also be divided into layers, with operations in each layer done in parallel, and the different layers applied in sequence. The second representation is used to derive independent generators for the normalizer and a canonical choice for the logical operators. Moreover, in a certain subfamily of intersecting subset codes, the encoding circuit simplifies to the same recursive pattern of CNOT gates that is used to encode classical Reed-Muller and polar codes. In this subfamily, the $X$ and $Z$ parity check matrices are related to a natural generalization of classical Reed-Muller codes, and this allows one to derive a formula for the quantum code distance. 

It is worth noting that while classical Reed-Muller codes have been used previously to construct quantum stabilizer codes \cite{knill1996threshold,steane1999quantumreedmuller,sarvepalli2009asymmetric}, the methods proposed there lead to codes with stabilizer weight exceeding the weight of logical operators. It is only through generalizing classical Reed-Muller codes that the present work is able to obtain examples with distance exceeding the syndrome measurement weight. 

\paragraph{Communication or fault-tolerance}

In the standard model for communications, encoding and decoding are considered as perfect operations, and errors come only from the noisy channel in between. This model is easiest to work with, but does not accurately capture the imperfections of quantum hardware. The vast majority of previous works on quantum LDPC codes focus on this case. 

On the other end of the specturm is the standard model for fault-tolerance. Here, a specific compilation of encoding and decoding in terms of elementary gates must be considered, and each component of the circuits introduces errors. This model is very difficult to work with and consequently there are very few works on quantum LDPC codes that deal directly with fault tolerance. Remarkable is the work \cite{gottesman2014fault}, which shows that if a quantum LDPC family that satisfies certain condition exists, then fault tolerant quantum computation is possible with asymptotically constant overhead. Later, \cite{fawzi2018constant} showed that a hypergraph product \cite{tillich2014quantum} of classical expander codes \cite{sipser1996expander} satisfies the assumtions of \cite{gottesman2014fault}. More recently, \cite{bravyi2024highthreshold} has constructed a fault-tolerant quantum memory based on tensor product generalized bicycle codes \cite{kovalev2013quantum}. Another interesting recent work is \cite{delfosse2023spacetime}, which reduces the problem of correcting faults in a circuit of Clifford gates and Pauli measurements to the problem of correcting errors in a stabilizer code and identifies conditions under which the resulting code is LDPC. 

A central issue in the study of fault-tolerance with quantum LDPC codes is the ability to correct errors even when syndrome extraction is noisy \cite{gottesman2014fault,fawzi2018constant,bravyi2024highthreshold}. Recently, an intermediate error model has received attention \cite{ashikhmin2020quantumdatasyndrome,nemec2023quantumdatasyndrome}, in which syndrome outcomes are noisy but the full error propagation in the syndrome measurement circuit is not considered. The advantage of this intermediate model is that it is easier to work with than the full-circuit error model, while still giving some insight in the fault-tolerance properties of the studied error correcting codes. 

Along the axis communication versus fault-tolerance, the present work falls at the intermediate stage of considering syndrome errors. The parity check matrices of intersecting subset codes have linearly dedendent rows. Correctly extracted syndromes always fall in the images of the parity check matrices, which can be viewed as classical linear error correcting codes. The distances of these spaces of valid syndromes are measures of the ability to correct errors in the syndrome. These distances can be computed explicitly for a large subfamily of the general construction, owing to the connection to generalized Reed-Muller codes mentioned earlier. 

Moreover, the present work takes an initial step towards fault-tolerance, by providing an explicit compilation of encoding and syndrome measurements into elementary gates. In general, compilation into a circuit where many gates are performed in parallel is considered favorable for fault-tolerance, because the resulting total circuit depth is lower and there is less time for errors to accumulate. The special structure of intersecting subset codes implies that certain parallelization in the circuits for encoding and syndrome measurements is possible. 

\paragraph{Finite block length versus asymptotic regime}

The focus of much recent work has been on properties of asymptotic families of LDPC codes. However, it has been noted that asymptotic constructions do not necessarily produce the best quantum LDPC codes for small or medium block lengths \cite{panteleev2021degeneratequantum}. As a concrete example of this observation relevant to the present work, \cite{ostrev2024classicalproduct} shows that a $[[512,174,8]]$ quantum CSS code obtained from the 3D classical product of the $[8,7,2]$ single parity check code has much better empirical performance than a quantum Tanner code with similar block size and rate.

The present work focuses on block sizes of a few tens or a few hundreds of qubits. There are two reasons for this choice. First, these numbers reflect the current limitations of quantum hardware. Second, modern classical communications systems continue to use these block sizes in some cases, even though there now exists classical hardware capable of handling much longer error correcting codes. Therefore, quantum LDPC codes with small and medium block sizes are interesting at present and are likely to remain so even in a hypothetical future with much better quantum hardware.

The rest of the paper is structured as follows. Section \ref{sec:Preliminaries} covers notational conventions, background material, and some preparatory technical results. Then, section \ref{sec:GeneralCase} gives the general case of intersecting subset codes and derives various properties. Section \ref{sec:SpecialCase} considers a subfamily, defines a generalization of classical Reed Muller codes, and uses these to derive a formula for the distance of the quantum codes as well as the distances of the classical linear codes protecting against syndrome errors. Section \ref{sec:Examples} contains an extensive collection of examples with block sizes $2^m$, $m=3, \dots, m$ and with syndrome measurements on 2, 4, or 8 qubits. Section \ref{sec:Conclusion} concludes the paper and gives some possible directions for future work.  

\section{Preliminaries}\label{sec:Preliminaries}

\subsection{Conventions for row and column indexing}

For an $l \times n$ matrix $A$, the rows will be indexed by
$
[l]=\{0,1,\dots,l-1\}
$
and the columns will be indexed by
$
[n]=\{0,1,\dots,n-1\}
$
Note that sometimes $[n]$ is used to denote the set $\{1,\dots,n\}$, but in the present paper it will be more convenient to start at 0. 

For a tuple of matrices $A_i$ with respective sizes $l_i \times n_i$, rows and columns of 
$
A_0 \otimes A_1 \otimes \dots \otimes A_v 
$
will be indexed resepctively by 
$
[\vec{l}\;]=[l_0] \times \dots \times [l_v]\
$
and
$
[\vec{n}]=[n_0] \times \dots \times [n_v]
$
.

If a linear indexing of rows and columns of $A_1 \otimes \dots \otimes A_v$ is desired, this will be done lexicographically. Thus, the block decomposition
\begin{equation}
\twobytwo{A}{B}{C}{D} \otimes E = \twobytwo{A\otimes E}{B \otimes E}{C \otimes E}{D \otimes E}
\end{equation}
holds. 

\subsection{Conventions for subsets and indicator vectors}\label{sec:ConventionsForSubsetsAndIndicatorVectors}

To $S \subset [n]$ is associated a $1 \times n$ indicator vector that has $1$ in position $i$ if $i \in S$ and zero otherwise. The indicator vector will also be denoted by $S$. 

To a tuple
\begin{equation}
\mathbf{S}=\threebyone{S_0}{\vdots}{S_{u-1}}
\end{equation}
of subsets of $[n]$ is associated the $u \times n$ matrix of zeros and ones with rows $S_0,\dots,S_{u-1}$. This indicator matrix will also be denoted by $\mathbf{S}$. 

\subsection{Incomplete permutation matrices}\label{sec:IncompletePermutationMatrices}

Take
$
l,n \in\N, S \subset [n]
$, 
injective 
$
\alpha:S\rightarrow [l]
$
To $l,n,S,\alpha$ associate the matrix
$
\GP{S} \in \F^{l\times n}
$
that has one in positions $(\alpha(i),i),i\in S$ and zero everywhere else. To keep the notation simple, $l,n,\alpha$ are left implicit. One may call $\GP{S}$ an incomplete permutation matrix with column support $S$. 

An immediate consequence of the definition is
$
\GP{S_1} \otimes \GP{S_2} = \GP{S_1 \times S_2}
$.
Indeed, both $\GP{S_1} \otimes \GP{S_2}$ and $\GP{S_1 \times S_2}$ are zero everywhere except in positions with row index $(\alpha_1(i),\alpha_2(j))$ and column index $(i,j)$ for $i\in S_1, j \in S_2$. 

Besides this, the following will also be useful:
\begin{theorem}\label{thm:LayersOfIncompletePermutations}
Take
$
n,l_0,\dots,l_k\in\N
$
and let
$l=\sum_{i=0}^k l_i$. 
Take
$
S_0, \dots S_k \subset [n]  
$
and let $T=\cup_{i=0}^k S_i$. 
Take injective
$
\alpha_i:S_i\rightarrow [l_i]
$ and $\alpha:T\rightarrow [l]$.
Then, there exists invertible $\Lambda(\mathbf{S}) \in \F^{l \times l}$ such that 
\begin{equation}
\threebyone{\GP{S_0}}{\vdots}{\GP{S_k}}=\Lambda(\mathbf{S})\GP{T}
\end{equation}
where, to keep the notation simple, the dependence of $\Lambda$ on
\newline 
$
n,l_0,\dots,l_k,\alpha_0,\dots,\alpha_k,\alpha
$
is left implicit. 
\end{theorem}

\begin{proof}
Let
$
a_0 < a_1 < \dots < a_r
$
be the elements of $T$. 
For $j=0, \dots, r$, let
$
b_j=|\{i:a_j\in S_i\}|
$
be the number of subsets in which $a_j$ appears. 

The matrix $\Lambda(\mathbf{S})$ is obtained as a product of two permutation matrices and a lower triangular matrix:
\begin{enumerate}
\item The first permutation matrix shuffles the rows of $\GP{T}$ so that the non-zero elements are in positions $(j,a_j),j=0,\dots,r$. 
\item The lower triangular matrix copies $(b_j-1)$ times row $j$ for $j=0,\dots,r$. 
\item The second permutation matrix shuffles the rows to their correct final positions, specified by the injections $\alpha_0,\dots, \alpha_k$ and the order of the sets $S_0,\dots, S_k$.
\end{enumerate}
\end{proof}

\textbf{Example:} Let $n=4$, $l_0=l_1=2$. Then, $l=4$. Let the tuple $\mathbf{S}$ consist of $S_0=\{0,1\}$ and $S_1=\{0,2\}$. Then, $T=\{0,1,2\}$. Let the injection $\alpha_0:S_0\rightarrow[2]$ be $\alpha_0(0)=0, \alpha_0(1)=1$. Then, the incomplete permutation matrix corresponding to $S_0,\alpha_0$ is 
\begin{equation}
\GP{S_0}=\twobyfour{1}{0}{0}{0}{0}{1}{0}{0}
\end{equation}
Let the injection $\alpha_1:S_1\rightarrow[2]$ be $\alpha_1(0)=0, \alpha_1(2)=1$. Then, the incomplete permutation matrix corresponding to $S_1,\alpha_1$ is 
\begin{equation}
\GP{S_1}=\twobyfour{1}{0}{0}{0}{0}{0}{1}{0}
\end{equation}
Let the injection $\alpha:T\rightarrow[4]$ be $\alpha(0)=0,\alpha(1)=1,\alpha(2)=2$. Then, the incomplete permutation matrix corresponding to $T,\alpha$ is 
\begin{equation}
\GP{T}=
\begin{pmatrix}
1&0&0&0\\
0&1&0&0\\
0&0&1&0\\
0&0&0&0
\end{pmatrix}
\end{equation}
Applying the above argument to $\GP{S_0},\GP{S_1},\GP{T}$ gives $a_0=0,a_1=1,a_2=2$ as the elements of $T$ in increasing order, $b_0=2,b_1=1,b_2=1$ as the number of subsets in which they appear, and
\begin{equation}
\twobyone{\GP{S_0}}{\GP{S_1}} = 
\begin{pmatrix}
1&0&0&0\\
0&1&0&0\\
0&0&0&1\\
0&0&1&0
\end{pmatrix}
\begin{pmatrix}
1&0&0&0\\
0&1&0&0\\
0&0&1&0\\
1&0&0&1
\end{pmatrix}
\begin{pmatrix}
1&0&0&0\\
0&1&0&0\\
0&0&1&0\\
0&0&0&1
\end{pmatrix}
\GP{T}
\end{equation}
as the two permutation matrices and the lower triangular matrix that relate $\GP{T}$ to $\GP{S_0},\GP{S_1}$.

\subsection{Joint decomposition of a pair of orthogonal matrices}\label{sec:JointDecompositionForPairOfOrthogonalMatrices}

An important tool in linear algebra is the decomposition of a matrix into a product of two permutation matrices, a lower triangular, an upper triangular, and a diagonal matrix. This section shows that a joint decomposition exists for a pair of orthogonal matrices. 

\begin{theorem}\label{thm:JointDecompositionForPairOfOrthogonalMatrices}
Take
$
A \in \F^{m_a \times n}, B \in \F^{m_b \times n}
$
such that $AB^T=0$. 
Denote the ranks of $A,B$ by $r_a,r_b$. 
Then, there exist
\begin{enumerate}
\item permutation matrices $P_a \in \F^{m_a \times m_a},P_b \in \F^{m_b \times m_b},Q \in \F^{n \times n}$, 
\item lower triangular $L_a\in \F^{m_a \times m_a}$ with ones along the diagonal.
\item upper triangular $L_b \in \F^{m_b \times m_b},R\in\F^{n \times n}$ with ones along the diagonal. 
\item $D_a \in \F^{m_a \times n}$ that has a one in the $(i,i)$ entry for $i=0,\dots r_a-1$ and zero everywhere else. 
\item $D_b \in \F^{m_b \times n}$ that has a one in the $(m_b-r_b+i,n-r_b+i)$ entry for $i=0,\dots r_b-1$ and zero everywhere else. 
\end{enumerate}
such that 
\begin{align}
P_aAQ &= L_a D_a R \\
P_bBQ &= L_b D_b R^{-T}
\end{align}
where $R^{-T}=(R^{-1})^T$ is the transpose of the inverse of $R$. 
\end{theorem}

\begin{proof}
If both $A,B$ are zero, the theorem holds trivially. 

If one of $A,B$ is zero, the claim follows from the usual decomposition of the non-zero matrix. 

If both $A,B$ are non-zero, then take $i,j$ such that entry $i,j$ of $A$ is one, and take $k$ such that row $k$ of $B$ is non-zero. The assumption $AB^T=0$ implies row $i$ of $A$ and row $k$ of $B$ are orthogonal. Then, there exists $l \neq j$ such that entry $k,l$ of $B$ is one.
Then, there exist permutation matrices $P_a,P_b,Q$ such that $P_aAQ$ and $P_bBQ$ have the block form
\begin{align}
P_aAQ &= \twobythree{1}{A_{12}}{A_{13}}{A_{21}}{A_{22}}{A_{23}} \\
P_bBQ &= \twobythree{B_{11}}{B_{12}}{B_{13}}{B_{21}}{B_{22}}{1}
\end{align}
where the number of rows in the blocks is $(1,m_a-1)$ for $A$ and $(m_b-1,1)$ for $B$, and the number of columns in the blocks is $(1,n-2,1)$ for both. 

Now take 
\begin{equation}
R=\threebythree{1}{A_{12}}{B_{21}}{0}{I}{B_{22}^T}{0}{0}{1}
\end{equation}
where the number of rows and of columns in the blocks is $(1,n-2,1)$. 

Note that $AB^T=0$ implies
\begin{equation}
0 =P_aAQQ^TB^TP_b^T = \twobythree{1}{A_{12}}{A_{13}}{A_{21}}{A_{22}}{A_{23}} \threebytwo{B_{11}^T}{B_{21}^T}{B_{12}^T}{B_{22}^T}{B_{13}^T}{1}
\end{equation}
The resulting relations for the blocks $A_{ij},B_{kl}$ can be used to verify that
\begin{align}
R^{-T} &= \threebythree{1}{0}{0}{A_{12}^T}{I}{0}{A_{13}}{B_{22}}{1} \\
P_aAQR &= \twobythree{1}{0}{0}{A_{21}}{A_{21}A_{12}+A_{22}}{0} \\
P_bBQR^{-T} &=\twobythree{0}{B_{12}+B_{13}B_{22}}{B_{13}}{0}{0}{1}
\end{align}

Now, take
\begin{align}
L_a&=\twobytwo{1}{0}{A_{21}}{I} \\
L_b &= \twobytwo{I}{B_{13}}{0}{1}
\end{align}
where the number of rows and columns in the blocks is $(1,m_a-1)$ for $L_a$ and $(m_b-1,1)$ for $L_b$. 

Note that
\begin{align}
L_aP_aAQR &= \twobythree{1}{0}{0}{0}{A_{21}A_{12}+A_{22}}{0} \\
L_bP_bBQR^{-T} &=\twobythree{0}{B_{12}+B_{13}B_{22}}{0}{0}{0}{1}
\end{align}

From here, the proof can be completed by induction on $max(m_a,m_b,n)$. First, note that $$L_aP_aAQR(L_bP_bBQR^{-T})^T=L_aP_aAB^TP_b^TL_b^T=0 $$ and therefore $(A_{21}A_{12}+A_{22})(B_{12}+B_
{13}B_{22})^T=0$. By the induction hypothesis, 
\begin{multline}
L_aP_aAQR=\twobytwo{1}{0}{0}{{P_a'}^{-1}} \twobytwo{1}{0}{0}{L_a'}\twobythree{1}{0}{0}{0}{D_a'}{0}\\ *\threebythree{1}{0}{0}{0}{R'}{0}{0}{0}{1}\threebythree{1}{0}{0}{0}{{Q'}^{-1}}{0}{0}{0}{1}
\end{multline}
\begin{multline}
L_bP_bBQR^{-T}=\twobytwo{{P_b'}^{-1}}{0}{0}{1}\twobytwo{L_b'}{0}{0}{1}\twobythree{0}{D_b'}{0}{0}{0}{1}\\ *\threebythree{1}{0}{0}{0}{{R'}^{-T}}{0}{0}{0}{1}\threebythree{1}{0}{0}{0}{{Q'}^{-1}}{0}{0}{0}{1}
\end{multline}
for suitable $L_a',L_b',R',P_a',P_b',Q',D_a',D_b'$.

Now, move $P_a',P_b',Q'$ to the other side, past $L_a,L_b,R,R^{-T}$ and merge them with $P_a,P_b,Q$. Moving $P_a',P_b',Q'$ past $L_a,L_b,R,R^{-T}$ is acheived using the relations:
\begin{align}
\twobytwo{1}{0}{0}{P'_a}\twobytwo{1}{0}{A_{21}}{I} &= \twobytwo{1}{0}{P'_aA_{21}}{I}\twobytwo{1}{0}{0}{P'_a} \\
\twobytwo{P'_b}{0}{0}{1}\twobytwo{I}{B_{13}}{0}{1} &= \twobytwo{I}{P'_bB_{13}}{0}{1}\twobytwo{P'_b}{0}{0}{1} \\
\threebythree{1}{A_{12}}{B_{21}}{0}{I}{B_{22}^T}{0}{0}{1}\threebythree{1}{0}{0}{0}{Q'}{0}{0}{0}{1} &= \threebythree{1}{0}{0}{0}{Q'}{0}{0}{0}{1}\threebythree{1}{A_{12}Q'}{B_{21}}{0}{I}{(Q')^{-1}B_{22}^T}{0}{0}{1} \\
\threebythree{1}{0}{0}{A_{12}^T}{I}{0}{A_{13}}{B_{22}}{1}\threebythree{1}{0}{0}{0}{Q'}{0}{0}{0}{1} &= \threebythree{1}{0}{0}{0}{Q'}{0}{0}{0}{1}\threebythree{1}{0}{0}{(Q')^{-1}A_{12}^T}{I}{0}{A_{13}}{B_{22}Q'}{1} 
\end{align}
Finally, move the so transformed $L_a,L_b,R,R^{-T}$ to the other side and merge them with $L_a',L_b',R',{R'}^{-T}$. This completes the inductive step and the proof. 
\end{proof}

\subsection{Matrices with layers of tensor products and their decompositions}\label{sec:MatricesWithLayersOfTensorProducts}

Take a tuple of matrices
$
\mathbf{H}=(H_0,\dots,H_{m-1}), \;\; H_i \in \F^{l_i \times n_i},i=0,\dots,m-1
$.
Take a subset $X \subset [m]$. To the pair $\mathbf{H},X$ associate the matrix 
\begin{equation}
M(\mathbf{H},X)=\otimes_{i=0}^{m-1}
\begin{cases}
H_i & \text{if } i \in X\\
I_{n_i} & \text{otherwise}
\end{cases}
\end{equation}

\textbf{Example:} $m=2, H_0=H_1=\onebytwo{1}{1}, X=\{0\}$, 
\begin{equation}
M(\mathbf{H},X)=\onebytwo{1}{1} \otimes \twobytwo{1}{0}{0}{1} = \twobyfour{1}{0}{1}{0}{0}{1}{0}{1}
\end{equation}

Now, suppose that tuples of permutation matrices $\mathbf{P},\mathbf{Q}$, tuples of invertible matrices $\mathbf{L},\mathbf{R}$, and a tuple of diagonal matrices $\mathbf{D}$ are known such that the decompositions
$
P_iH_iQ_i=L_iD_iR_i, \;\; i=0,\dots,m-1
$
hold. From these, obtain the decomposition
\begin{equation}\label{eq:LDUforSingleLayer}
P(\mathbf{H},X)M(\mathbf{H},X)Q(\mathbf{H})=L(\mathbf{H},X)D(\mathbf{H},X)R(\mathbf{H})
\end{equation}
where
\begin{align}
P(\mathbf{H},X) &= \otimes_{i=0}^{m-1}
\begin{cases}
P_i & \text{if } i \in X\\
Q_i^{-1} & \text{otherwise}
\end{cases} \\
Q(\mathbf{H}) &= \otimes_{i=0}^{m-1} Q_i \\
L(\mathbf{H},X) &= \otimes_{i=0}^{m-1}
\begin{cases}
L_i & \text{if } i \in X\\
R_i^{-1} & \text{otherwise}
\end{cases}\\
D(\mathbf{H},X) &= \otimes_{i=0}^{m-1}
\begin{cases}
D_i & \text{if } i \in X\\
I_{n_i} & \text{otherwise}
\end{cases}\\
R(\mathbf{H}) &= \otimes_{i=0}^{m-1} R_i
\end{align}
To keep the notation simple, assume that the particular decompositions used for the matrices in $\mathbf{H}$ are known from context and do not need to be specified explicitly. 

\textbf{Example:} Take the decomposition
\begin{equation}
\onebytwo{1}{1}=1\onebytwo{1}{0}\twobytwo{1}{1}{0}{1}
\end{equation}
and take $m=2, H_0=H_1=\onebytwo{1}{1}, X=\{0\}$ as before. Then,
\begin{multline}
M(\mathbf{H},X)=\onebytwo{1}{1} \otimes \twobytwo{1}{0}{0}{1} \\= \left(1 \otimes \twobytwo{1}{1}{0}{1}\right)\left(\onebytwo{1}{0}\otimes\twobytwo{1}{0}{0}{1}\right)\left(\twobytwo{1}{1}{0}{1}\otimes\twobytwo{1}{1}{0}{1}\right)
\end{multline}

Now, take a tuple of matrices $\mathbf{H}$ as before, but this time take a tuple of subsets of $[m]$
\begin{equation}
\mathbf{X}=\threebyone{X_0}{\vdots}{X_{k-1}},\;\;X_i \subset [m], i=0,\dots,k-1
\end{equation}
To the pair $\mathbf{H},\mathbf{X}$, associate the matrix
\begin{equation}\label{eq:LayeredTensorProductMatrixNotation}
M(\mathbf{H},\mathbf{X})=\threebyone{M(\mathbf{H},X_0)}{\vdots}{M(\mathbf{H},X_{k-1})}
\end{equation}
obtained by taking $M(\mathbf{H},X_0), \dots, M(\mathbf{H},X_{k-1})$ as layers of rows. 

The matrix $M(\mathbf{H},\mathbf{X})$ when $\mathbf{X}$ is a tuple of subsets can also be decomposed as
\begin{equation}\label{eq:LDUforLayersOfTensorProducts}
P(\mathbf{H},\mathbf{X})M(\mathbf{H},\mathbf{X})Q(\mathbf{H})=L(\mathbf{H},\mathbf{X})D(\mathbf{H},\mathbf{X})R(\mathbf{H})
\end{equation}
where
\begin{align}
P(\mathbf{H},\mathbf{X}) &= \threebythree{P(\mathbf{H},X_0)}{}{}{}{\ddots}{}{}{}{P(\mathbf{H},X_{k-1})} \\
Q(\mathbf{H}) &= \otimes_{i=0}^{m-1} Q_i \\
L(\mathbf{H},\mathbf{X}) &= \threebythree{L(\mathbf{H},X_0)}{}{}{}{\ddots}{}{}{}{L(\mathbf{H},X_{k-1})}\Lambda(\mathbf{S})\\
D(\mathbf{H},\mathbf{X}) &= \GP{T}\\
R(\mathbf{H}) &= \otimes_{i=0}^{m-1} R_i
\end{align}
and where $S_i$ is the column support of $D(\mathbf{H},X_i)$ viewed as an incomplete permutation matrix, $T=\cup_i S_i$, and $\Lambda(\mathbf{S})$ is defined in Theorem \ref{thm:LayersOfIncompletePermutations} in section  \ref{sec:IncompletePermutationMatrices}. 

\textbf{Example:} Take the decomposition
\begin{equation}
\onebytwo{1}{1}=1\onebytwo{1}{0}\twobytwo{1}{1}{0}{1}
\end{equation}
and take $m=2, H_0=H_1=\onebytwo{1}{1}$ as before. Take $X_0=\{0\}, X_1=\{1\}$. The matrix $M(\mathbf{H},\mathbf{X})$ has two layers; decomposing each layer separately gives
\begin{equation}
M(\mathbf{H},\mathbf{X})=\twobyone{M(\mathbf{H},X_0)}{M(\mathbf{H},X_1)}=\twobytwo{L(\mathbf{H},X_0)}{0}{0}{L(\mathbf{H},X_1)}\twobyone{D(\mathbf{H},X_0)}{D(\mathbf{H},X_1)}R(\mathbf{H})
\end{equation} 
The matrix 
\begin{equation}
D(\mathbf{H},X_0) = \onebytwo{1}{0}\otimes\twobytwo{1}{0}{0}{1}=\twobyfour{1}{0}{0}{0}{0}{1}{0}{0}
\end{equation}
is an incomplete permutation matrix with column support $S_0=\{0,1\}$. The matrix 
\begin{equation}
D(\mathbf{H},X_1) = \twobytwo{1}{0}{0}{1}\otimes\onebytwo{1}{0}=\twobyfour{1}{0}{0}{0}{0}{0}{1}{0}
\end{equation}
is an incomplete permutation matrix with column support $S_1=\{0,2\}$. The example after Theorem \ref{thm:LayersOfIncompletePermutations} in section  \ref{sec:IncompletePermutationMatrices} shows that 
\begin{equation}
\twobyone{D(\mathbf{H},X_0)}{D(\mathbf{H},X_1)}=\begin{pmatrix}
1&0&0&0\\
0&1&0&0\\
1&0&0&1\\
0&0&1&0
\end{pmatrix}
\begin{pmatrix}
1&0&0&0\\
0&1&0&0\\
0&0&1&0\\
0&0&0&0
\end{pmatrix}
=\Lambda(\mathbf{S}) \GP{T}
\end{equation}
Thus, the decomposition of $M(\mathbf{H},\mathbf{X})$ from equation \eqref{eq:LDUforLayersOfTensorProducts} has $P(\mathbf{H},\mathbf{X})=Q(\mathbf{H})=I_4$, and
\begin{align}
L(\mathbf{H},\mathbf{X})&=\twobytwo{L(\mathbf{H},X_0)}{0}{0}{L(\mathbf{H},X_1)}\Lambda(\mathbf{S})=
\begin{pmatrix}
1&1&0&0\\
0&1&0&0\\
1&0&1&1\\
0&0&1&0
\end{pmatrix}\\
D(\mathbf{H},\mathbf{X})&=\GP{T}=\begin{pmatrix}
1&0&0&0\\
0&1&0&0\\
0&0&1&0\\
0&0&0&0
\end{pmatrix}\\
R(\mathbf{H})&=\twobytwo{1}{1}{0}{1}\otimes\twobytwo{1}{1}{0}{1}=\begin{pmatrix}
1&1&1&1\\
0&1&0&1\\
0&0&1&1\\
0&0&0&1
\end{pmatrix}
\end{align}

\subsection{The partial order on integer tuples}\label{sec:PartialOrder}

For $x,y \in \Z^n$,
$
x<y \Leftrightarrow \AND{\forall i, x_i \leq y_i}{\exists i, x_i < y_i}
$
and this defines a partial order on $\Z^n$; for example, $(1,2,3)<(2,3,4)<(3,4,5)$, but $(1,2,3)$ and $(4,3,2)$ are incomparable. 

Now, take $n_0,\dots,n_{k-1} \in \N$ and consider the set
$
[\vec{n}]=[n_0] \times [n_2] \times \dots \times [n_{k-1}]
$
This set inherits from $\Z^k$ the partial order. 

A subset $T \subset [\vec{n}]$ is called increasing if
$
\AND{x \in T}{x<y} \Rightarrow y \in T
$
and decreasing if 
$
\AND{x \in T}{y<x} \Rightarrow y \in T
$.The minimal elements of $T$ are 
$
\MIN{T}=\{x\in T: \forall y \in T, \NOT{y<x}\}
$
and the maximal elements of $T$ are 
$
\MAX{T}=\{x\in T: \forall y \in T, \NOT{x<y}\}
$.
The increasing subset generated by $T $ is 
$
\GUP{T}=\{x \in S:\exists y \in T,y \leq x\}
$
and the decreasing subset generated by $T$ is 
$
\GDO{T}=\{x \in S:\exists y \in T,x \leq y\}
$.

\subsection{Quantum stabilizer codes}

\subsubsection{The Pauli group}

The Pauli matrices are 
\begin{equation}
\PauliX=\twobytwo{0}{1}{1}{0} \;\; \PauliY=\twobytwo{0}{-i}{i}{0} \;\; \PauliZ=\twobytwo{1}{0}{0}{-1}
\end{equation}
The Pauli group on $n$ qubits is
$$
\PauliGroup{n}=\left\{w \PauliX^{u_1}\PauliZ^{v_1} \otimes \dots \otimes \PauliX^{u_n}\PauliZ^{v_n}:w\in\{\pm1,\pm i\},u,v \in \F^{1 \times n} \right\}
$$
The shorthand notation
$
\PauliX^u\PauliZ^v=\PauliX^{u_1}\PauliZ^{v_1} \otimes \dots \otimes \PauliX^{u_n}\PauliZ^{v_n}
$
will be used below. 

The map
$
\PauliGroupHomomorphism :\PauliGroup{n}\rightarrow\F^{1\times 2n},  \PauliGroupHomomorphism(w\PauliX^u\PauliZ^v)=\onebytwo{u}{v}
$
is a surjective group homomorphism with kernel $\{\pm I, \pm iI\}$. 
Two elements $w\PauliX^u\PauliZ^v,w'\PauliX^{u'}\PauliZ^{v'}$ of the Pauli group anti-commute if 
$
\PauliGroupHomomorphism(w\PauliX^u\PauliZ^v)\SymplecticProductMatrix\PauliGroupHomomorphism(w'\PauliX^{u'}\PauliZ^{v'})^T=1
$
and commute if 
$
\PauliGroupHomomorphism(w\PauliX^u\PauliZ^v)\SymplecticProductMatrix\PauliGroupHomomorphism(w'\PauliX^{u'}\PauliZ^{v'})^T=0
$
where
\begin{equation}
\SymplecticProductMatrix=\twobytwo{0}{I}{I}{0}
\end{equation}

\subsubsection{Stabilizer codes}
An abelian subgroup $\mathcal{S}$ of $\PauliGroup{n}$ that does not contain $-I$ is called a stabilizer subgroup. The joint $+1$-eigenspace of the elements of $\mathcal{S}$ is a quantum error correcting code. If elements $s_1,\dots,s_m \in \mathcal{S}$ generate $\mathcal{S}$, then the matrix $H$ with rows $\PauliGroupHomomorphism(s_1),\dots,\PauliGroupHomomorphism(s_m)$ is a parity check matrix for the code stabilized by $\mathcal{S}$. If a quantum state is encoded in the subspace stablized by $\mathcal{S}$, a Pauli error $E \in \PauliGroup{n}$ occurs, and $s_1, \dots, s_m$ are measured, then the resulting syndrome is $H\SymplecticProductMatrix \PauliGroupHomomorphism(E)^T$. 

Conversely, any matrix $H$ such that $H\SymplecticProductMatrix H^T=0$ is the parity check matrix of some stabilizer code. 

\subsubsection{The Clifford group and the symplectic group}

The Clifford group on $n$ qubits is the normalizer of the Pauli group in the $2^n\times 2^n$ unitary matrices:
$
\CliffordGroup{n}=\Normalizer{\PauliGroup{n}}{\UnitaryGroup{2^n}}=\left\{U\in\UnitaryGroup{2^n}:U\PauliGroup{n}=\PauliGroup{n}U\right\}
$.
Each Clifford group element can be expressed as a product of CNOT, Hadamard and Phase gates and a global phase. 

The symplectic group consists of matrices which preserve the symplectic product:
$
\SymplecticGroup{n} = \left\{ A \in \F^{2n \times 2n}:A \SymplecticProductMatrix A^T=\SymplecticProductMatrix \right\}
$.

There is a surjective group homomorphism
$
\CliffordGroupHomomorphism : \CliffordGroup{n}\rightarrow\SymplecticGroup{n}
$
with the property:
\begin{equation}\label{eq:PropertyOfCliffordGroupHomomorphism}
\forall P \in \PauliGroup{n}, \; \forall U \in \CliffordGroup{n}, \;\;\; \PauliGroupHomomorphism(U^{-1}PU)=\PauliGroupHomomorphism(P)\CliffordGroupHomomorphism(U)
\end{equation}

\subsubsection{Encoding circuits for stabilizer codes}

For $k \leq n \in \N$, and for $U \in \CliffordGroup{n}$, the Pauli group elements $U^{\dagger} {\PauliZ}_i U, i=1,\dots,n-k$ generate a stabilizer subgroup of $\PauliGroup{n}$, where ${\PauliZ}_i$ denotes $\PauliZ$ acting on the $i$-th qubit. 

Conversely, for each stabilizer subgroup $\mathcal{S}$ of $\PauliGroup{n}$, there exist many $U \in \CliffordGroup{n}$ such that $U^\dagger {\PauliZ}_i U, i=1,\dots,n-k$ generate $\mathcal{S}$. Given one such $U$, an encoding circuit for the code stabilized by $\mathcal{S}$ is the following:
\begin{enumerate}
\item Prepare qubits $1,\dots n-k$ in the state $\ket{0}$ and prepare the remaining $k$ qubits in the state to be encoded.
\item Apply a circuit of CNOT, Hadamard and Phase gates that realizes $U^\dagger$ (up to global phase). 
\end{enumerate}
The quantum code stabilized by $\mathcal{S}$ encodes $k$ qubits in $n$ qubits and is therefore called an $[[n,k]]$ code. 

\subsubsection{Normalizer, distance, logical operators}

The normalizer of a stablizer subgroup $\mathcal{S}$ in $\PauliGroup{n}$ is 
$$
\Normalizer{\mathcal{S}}{\PauliGroup{n}}=\left\{P\in\PauliGroup{n}:P\mathcal{S}=\mathcal{S}P\right\}
$$
The distance of the quantum code stabilized by $\mathcal{S}$ is the lowest weight of an element of $\Normalizer{\mathcal{S}}{\PauliGroup{n}}\backslash \left\langle {iI},\mathcal{S} \right\rangle$, where weight means the number of non-identity elements in the tensor product of Pauli matrices. 

If $U \in \CliffordGroup{n}$ is such that $U^\dagger {\PauliZ}_1 U, \dots, U^\dagger {\PauliZ}_{n-k} U$ generate $\mathcal{S}$, then $U^\dagger {\PauliZ}_i U, i=1, \dots, n$, $U^\dagger {\PauliX}_i U, i=n-k+1,\dots, n$ and $iI$ generate $\Normalizer{\mathcal{S}}{\PauliGroup{n}}$. The Pauli group elements $U^\dagger {\PauliX}_i U, U^\dagger {\PauliZ}_i U, i=n-k+1,\dots,n$ are one possible choice of logical operators for the quantum code stablized by $\mathcal{S}$. 

\subsubsection{Quantum CSS codes}

If a stabilizer subgroup has a set of generators such that each generator contains either only $\PauliX$ and $I$ terms or only $\PauliZ$ and $I$ terms, then the associated quantum code is called a Calderbank-Shor-Steane (CSS) code. A CSS code has a parity check matrix with the block diagonal form 
\begin{equation}
H=\twobytwo{H^x}{0}{0}{H^z}
\end{equation}
where $H^x(H^z)^T=0$. Conversely, for every pair of matrices $H^x, H^z$ such that $H^x(H^z)^T=0$, there is a CSS code that has the parity check matrix $\twobytwo{H^x}{0}{0}{H^z}$. 

A CSS code has an encoding circuit in the following form:
\begin{enumerate}
\item Prepare qubits $1, \dots, m_x$ in the state $\ket{+}$, prepare qubits $m_x+1,\dots,m_x+m_z$ in the state $\ket{0}$, prepare the remaining qubits in the state to be encoded. 
\item Apply a circuit consisting of only CNOT gates. 
\end{enumerate}

To a quantum CSS code are associated two distances
\begin{align}
d_x &= \min\left\{|v|:v\in Ker(H^z), v \notin Im((H^x)^T)\right\} \\
d_z &= \min \left\{|v|:v \in Ker(H^x),v \notin Im((H^z)^T) \right\}
\end{align}
which capture the ability of the code to correct, respectively, $\PauliX$ and $\PauliZ$ errors. The distance of the CSS code when all Pauli errors are considered together is $d=\min(d_x,d_z)$. 

\section{Large CSS code from a tuple of smaller CSS codes and two collections of subsets}\label{sec:GeneralCase}

This section begins by associating a quantum CSS code to two tuples of matrices and two tuples of subsets (subsection \ref{subsec:Definition}). Then, a number of properties are established. Matrix decompositions of the parity check matrices of the code are derived in subsection \ref{sec:MatrixDecompositionCSS}. From these matrix decompositions can be computed the block size and rate (subsection \ref{subsec:BlockSizeAndRate}), a basis for the stabilizer (subsection \ref{subsec:BasisForStabilizer}), an encoding circuit (subsection \ref{subsec:EncodingCircuit}), the normalizer and the logical operators (subsection \ref{subsec:NormalizerLogicalOperators}) and the linear error-correcting code protecting against syndrome errors (subsection \ref{subsec:ProtectionAgainstSyndromeErrors}). Moreover, the syndrome measurements can be parallelized (subsection \ref{subsec:SyndromeMeasurementCircuit}), and the number of qubits in each measurement can be kept low by suitable choices of the components and subsets (subsection \ref{subsec:RowAndColumnWeight}). Throughout the section, the general results are illustrated by a particular small example. 

\subsection{Definition of the CSS code}\label{subsec:Definition}
\begin{definition}
Take two tuples of matrices
\begin{align}
\mathbf{H^x}&=(H^x_0, \dots H^x_{m-1}), H^x_i \in \F^{l^x_i \times n_i}, i=0,\dots,m-1 \\
\mathbf{H^z}&=(H^z_0, \dots H^z_{m-1}), H^z_i \in \F^{l^z_i \times n_i}, i=0,\dots,m-1
\end{align}
such that 
\begin{equation}\label{eq:OrthogonalityCondition}
\forall i, H^x_i (H^z_i)^T=0
\end{equation}

Take two tuples of subsets of $[m]$:
\begin{equation}
\mathbf{X}= \threebyone{X_0}{\vdots}{X_{u-1}} \;\;\; \mathbf{Z}=\threebyone{Z_0}{\vdots}{Z_{v-1}}
\end{equation}
such that 
\begin{equation}\label{eq:FullyIntersectingCondition}
\forall i \forall j, X_i \cap Z_j \neq \emptyset
\end{equation}

To $\mathbf{H^x},\mathbf{H^z},\mathbf{X},\mathbf{Z}$, associate the CSS code with stabilizer parity check matrix
\begin{equation}
CSS(\mathbf{H^x},\mathbf{H^z},\mathbf{X},\mathbf{Z})=\twobytwo{M(\mathbf{H^x},\mathbf{X})}{0}{0}{M(\mathbf{H^z},\mathbf{Z})}
\end{equation}
where the notation of section \ref{sec:MatricesWithLayersOfTensorProducts} is used. 
\end{definition}
\begin{theorem}
$CSS(\mathbf{H^x},\mathbf{H^z},\mathbf{X},\mathbf{Z})$ is a valid stabilizer parity check matrix.
\end{theorem}
\begin{proof}
\eqref{eq:FullyIntersectingCondition} and \eqref{eq:OrthogonalityCondition} imply
\begin{equation}
\forall i \forall j, M(\mathbf{H^x},X_i)M(\mathbf{H^z},Z_j)^T=0
\end{equation}
\end{proof}

\textbf{Running example:} Take $m=4$. Take $\mathbf{H^x},\mathbf{H^z}$ so that for each $i \in [4]$, $H^x_i=H^z_i=\onebytwo{1}{1}$. Take
\begin{equation}
\mathbf{X}=\twobyone{X_0}{X_1}=\twobyone{\{0,1\}}{\{2,3\}} \;\;\; \mathbf{Z}=\twobyone{Z_0}{Z_1}=\twobyone{\{0,2\}}{\{1,3\}}
\end{equation}
The two parity check matrices are
\begin{align}
M(\mathbf{H^x},\mathbf{X})&=\twobyone{\onebytwo{1}{1}\otimes\onebytwo{1}{1}\otimes\twobytwo{1}{0}{0}{1}\otimes\twobytwo{1}{0}{0}{1}}{\twobytwo{1}{0}{0}{1}\otimes\twobytwo{1}{0}{0}{1}\otimes\onebytwo{1}{1}\otimes\onebytwo{1}{1}} \\
M(\mathbf{H^z},\mathbf{Z})&=\twobyone{\onebytwo{1}{1}\otimes\twobytwo{1}{0}{0}{1}\otimes\onebytwo{1}{1}\otimes\twobytwo{1}{0}{0}{1}}{\twobytwo{1}{0}{0}{1}\otimes\onebytwo{1}{1}\otimes\twobytwo{1}{0}{0}{1}\otimes\onebytwo{1}{1}}
\end{align}

\begin{figure}[H]
\centering
\caption{The code $CSS(\mathbf{H^x},\mathbf{H^z},\mathbf{X},\mathbf{Z})$ in the running example}
\label{fig:2dproduct}
\begin{subfigure}{\textwidth}
\centering
\caption{Non-zero entries of $M(\mathbf{H^x},\mathbf{X})$}
\label{fig:2dproductMX}
\begin{tikzpicture}[scale=0.5]
\foreach \x in {0,...,16} {\draw (\x,0) -- (\x,-8);}
\foreach \y in {0,...,8} \draw (0,-\y) -- (16, -\y) ;
\foreach \a in {0,1}
	\foreach \b in {0,1}
		\foreach \c in {0,1}
			\foreach \d in {0,1}
				{
					\tikzmath{\x=8*\a+4*\b+2*\c+\d;}
					\tikzmath{\y=-2*\c-\d;}
					\fill (\x,\y)--(\x+1,\y)--(\x+1,\y-1)--(\x,\y-1)--(\x,\y);
					\tikzmath{\y=-2*\a-\b-4;}
					\fill (\x,\y)--(\x+1,\y)--(\x+1,\y-1)--(\x,\y-1)--(\x,\y);
				}
\end{tikzpicture}
\end{subfigure}
\begin{subfigure}{\textwidth}
\centering
\caption{Non-zero entries of $M(\mathbf{H^z},\mathbf{Z})$}
\label{fig:2dproductMZ}
\begin{tikzpicture}[scale=0.5]
\foreach \x in {0,...,16} {\draw (\x,0) -- (\x,-8);}
\foreach \y in {0,...,8} \draw (0,-\y) -- (16, -\y) ;
\foreach \a in {0,1}
	\foreach \b in {0,1}
		\foreach \c in {0,1}
			\foreach \d in {0,1}
				{
					\tikzmath{\x=8*\a+4*\b+2*\c+\d;}
					\tikzmath{\y=-2*\b-\d;}
					\fill (\x,\y)--(\x+1,\y)--(\x+1,\y-1)--(\x,\y-1)--(\x,\y);
					\tikzmath{\y=-2*\a-\c-4;}
					\fill (\x,\y)--(\x+1,\y)--(\x+1,\y-1)--(\x,\y-1)--(\x,\y);
				}
\end{tikzpicture}
\end{subfigure}
\begin{subfigure}{\textwidth}
\centering
\caption{Tanner graph embedded in the torus. Qubits are denoted by $Q_{abcd}$ for $a,b,c,d \in \{0,1\}$. Checks are denoted by $C^{A_b}_{cd}$ for $A \in \{X,Z\},b,c,d\in\{0,1\}$, so $C^{X_0}_{00}$ is the $00$ row of $M(\mathbf{H^x}, X_0)$, etc.}
\label{fig:2dproducttannergraph}
\begin{tikzpicture}
\foreach \x in {-4,-2,0,2,4} {\draw (\x,-4)--(4,-\x);\draw (\x,4)--(-4,-\x);\draw (\x,4)--(4,\x);\draw (\x,-4)--(-4,\x);}
\foreach \a in {0,1}
	\foreach \b in {0,1}
		\foreach \c in {0,1}
			\foreach \d in {0,1}
				{
					\tikzmath{\x=-3+6*\a+2*\d-4*\a*\d;}
					\tikzmath{\y=-3+6*\b+2*\c-4*\b*\c;}
					\draw (\x,\y) node[fill=white]{$Q_{\a\b\c\d}$};
				}
\foreach \a in {0,1}
	\foreach \b in {0,1}
		{
			\tikzmath{\x=-2+4*\a;}
			\tikzmath{\y=-2+4*\b;}
			\draw (\x,\y) node[fill=white]{$C_{\a\b}^{X_1}$};
		}
\foreach \a in {0,1}
	\foreach \b in {0,1}
		{
			\tikzmath{\x=4-4*\a;}
			\tikzmath{\y=4-4*\b;}
			\draw (\x,\y) node[fill=white]{$C_{\b\a}^{X_0}$};
		}
\foreach \a in {0,1}
	\foreach \b in {0,1}
		{
			\tikzmath{\x=-2+4*\a;}
			\tikzmath{\y=4-4*\b;}
			\draw (\x,\y) node[fill=white]{$C_{\a\b}^{Z_1}$};
		}
\foreach \a in {0,1}
	\foreach \b in {0,1}
		{
			\tikzmath{\x=4-4*\a;}
			\tikzmath{\y=-2+4*\b;}
			\draw (\x,\y) node[fill=white]{$C_{\b\a}^{Z_0}$};
		}
\end{tikzpicture}
\end{subfigure}
\end{figure}

$M(\mathbf{H^x},\mathbf{X})$ is the parity check matrix of the classical 2D product of the $[4,3]$ single parity check code; $M(\mathbf{H^z},\mathbf{Z})$ is isomorphic to $M(\mathbf{H^x},\mathbf{X})$ by row and column permutation. A visualization of $M(\mathbf{H^x},\mathbf{X})$ and $M(\mathbf{H^z},\mathbf{Z})$ is given in Figures \ref{fig:2dproductMX} and \ref{fig:2dproductMZ}. 

The resulting [[16,2,4]] quantum CSS code has been previously considered in \cite{hivadi2018quantum,ostrev2024classicalproduct} from the point of view of using classical 2D products to construct quantum CSS codes. However, it turns out that this is also a toric code.\footnote{The author would like to thank the anonymous reviewer who pointed out that a particular tessellation of the torus also gives a [[16,2,4]] CSS code. It turns out that this is a coincindence not just of the parameters but of the code itself.} Figure \ref{fig:2dproducttannergraph} shows the Tanner graph of $CSS(\mathbf{H^x},\mathbf{H^z},\mathbf{X},\mathbf{Z})$ embedded in the torus. 

\subsection{Matrix decomposition}\label{sec:MatrixDecompositionCSS}

A matrix decomposition of $CSS(\mathbf{H^x},\mathbf{H^z},\mathbf{X},\mathbf{Z})$ follows from the observations in sections \ref{sec:IncompletePermutationMatrices}, \ref{sec:JointDecompositionForPairOfOrthogonalMatrices}, \ref{sec:MatricesWithLayersOfTensorProducts}, \ref{sec:PartialOrder}. From this decomposition follow many properties of the code.

For $\mathbf{H^x},\mathbf{H^z}$ as above, for $i=0,\dots,m-1$, let 
\begin{align}
P^x_i H^x_i Q_i &= L^x_i D^x_i R_i \label{eq:ComponentDecompositionX}\\
P^z_i H^z_i Q_i &= L^z_i D^z_i R_i^{-T} \label{eq:ComponentDecompositionZ}
\end{align}
be the joint decomposition of the pair $H^x_i,H^z_i$ from Theorem \ref{thm:JointDecompositionForPairOfOrthogonalMatrices}.

Using these component decompositions, let 
\begin{align}
P(\mathbf{H^x},\mathbf{X})M(\mathbf{H^x},\mathbf{X})Q(\mathbf{H^x}) &= L(\mathbf{H^x},\mathbf{X})D(\mathbf{H^x},\mathbf{X}) R(\mathbf{H^x}) \label{eq:DecompositionOfMHxX}\\
P(\mathbf{H^z},\mathbf{Z})M(\mathbf{H^z},\mathbf{Z})Q(\mathbf{H^z}) &= L(\mathbf{H^z},\mathbf{Z})D(\mathbf{H^z},\mathbf{Z}) R(\mathbf{H^z})  \label{eq:DecompositionOfMHzZ}
\end{align}
be the resulting decompositions of $M(\mathbf{H^x},\mathbf{X}),M(\mathbf{H^z},\mathbf{Z})$ from equation \eqref{eq:LDUforLayersOfTensorProducts} of section \ref{sec:MatricesWithLayersOfTensorProducts}. 

Recall from section \ref{sec:MatricesWithLayersOfTensorProducts} that $D(\mathbf{H^x},\mathbf{X})$ and $D(\mathbf{H^z},\mathbf{Z})$ are incomplete permutation matrices. The immediate goal is to compute their column support. To do this, recall that equation \eqref{eq:LDUforSingleLayer} in section \ref{sec:MatricesWithLayersOfTensorProducts} gives a matrix decomposition for individual layers $M(\mathbf{H^x},X_i),M(\mathbf{H^z},Z_j)$. The column support of $D(\mathbf{H^x},\mathbf{X})$ is the union of the column supports of $D(\mathbf{H^x},X_i)$, and the column support $D(\mathbf{H^z},\mathbf{Z})$ is the union of the column supports of $D(\mathbf{H^z},Z_j)$. 

The following theorem computes the relevant column supports:
\begin{theorem}\label{thm:ColumnSupports}
Let
\begin{align}
\mathbf{S^x} &= \ones{u} \ones{m}^T Diag(\vec{n}-\ones{m}) - \mathbf{X} Diag(\vec{n}-\vec{r}^x) \label{eq:DecreasingSetGenerators}\\
\mathbf{S^z} &= \mathbf{Z} Diag(\vec{n}-\vec{r}^z) \label{eq:IncreasingSetGenerators}
\end{align}
where $\ones{m}$ is the $m \times 1$ vector of ones, $\vec{n}=(n_0,\dots,n_{m-1})^T$, where $\vec{r}^x=(r^x_0,\dots, r^x_{m-1})^T$, $\vec{r}^z=(r^z_0,\dots,r^z_{m-1})^T$ are the vectors of ranks of the matrices in the tuples $\mathbf{H^x}, \mathbf{H^z}$, and where $Diag(\vec{w})$ is the diagonal matrix obtained from vector $\vec{w}$. Here, $\mathbf{X},\mathbf{Z}$ are interpreted as $u \times m, v \times m$ indicator matrices as in Section \ref{sec:ConventionsForSubsetsAndIndicatorVectors} and $\mathbf{S^x},\mathbf{S^z}$ are interpreted as $u \times m, v\times m$ matrices with each row $S^x_0,\dots,S^x_{u-1}$, $S^z_0,\dots,S^z_{v-1}$ specifying an element of $[\vec{n}]$. 

With $\mathbf{S^x},\mathbf{S^z}$ so defined, the following hold:
\begin{enumerate}
\item For $i=0,\dots,u-1$, the column support of $D(\mathbf{H^x},X_i)$ is $\GDO{\{S^x_i\}} =\{t \in [\vec{n}]:t \leq S^x_i\}$. 
\item The column support of $D(\mathbf{H^x},\mathbf{X})$ is $\GDO{\mathbf{S^x}}=\{t\in[\vec{n}]:\exists i, t \leq S^x_i\}$. 
\item For $j=0,\dots,v-1$, the column support of $D(\mathbf{H^z},Z_j)$ is $\GUP{\{S^z_j\}}=\{t\in[\vec{n}]:t\geq S^z_j\}$.
\item The column support of $D(\mathbf{H^z},\mathbf{Z})$ is $\GUP{\mathbf{S^z}}=\{t\in[\vec{n}]:\exists j, t \geq S^z_j\}$. 
\item $\GDO{\mathbf{S^x}}$ and $\GUP{\mathbf{S^z}}$ are disjoint. 
\end{enumerate}
\end{theorem}

\begin{proof}
\textbf{Part 1:} Take any $i\in[u]$. Recall that 
\begin{equation}
D(\mathbf{H^x},X_i)=\otimes_{k=0}^{m-1} 
\begin{cases}
D^x_k & \text{if } k \in X_i \\
I_{n_k} & \text{otherwise} 
\end{cases}
\end{equation}
Then, column $t \in[\vec{n}]$ of $D(\mathbf{H^x},X_i)$ is non-zero if and only if $t_k$ is in the column support of $D^x_k$ for $k \in X_i$. Recall further that the column support of $D^x_k$ is $0,\dots,r^x_k-1$. Thus, column $t \in[\vec{n}]$ of $D(\mathbf{H^x},X_i)$ is non-zero if and only if 
\begin{equation}
t_k \leq 
\begin{cases}
r^x_k-1 & \text{if } k \in X_i \\
n_k-1 & \text{otherwise}
\end{cases}
\end{equation}
This is the same as $t \leq S^x_i$. 

\textbf{Part 2} follows from Part 1 because the column support of $D(\mathbf{H^x},\mathbf{X})$ is the union of the column supports of $D(\mathbf{H^x},X_i)$ for $i\in[u]$. 

\textbf{Part 3} is similar to Part 1, except that the column support of $D^z_k$ is $n_k-r^z_k, \dots, n_k-1$, so column $t$ of $D(\mathbf{H^z},Z_j)$ is non-zero if and only if 
\begin{equation}
t_k \geq 
\begin{cases}
n_k-r^z_k & \text{if } k \in Z_j \\
0 & \text{otherwise}
\end{cases}
\end{equation}
This is the same as $t \geq S^z_j$. 

\textbf{Part 4} follows from Part 3 because the column support of $D(\mathbf{H^z},\mathbf{Z})$ is the union of the column supports of $D(\mathbf{H^z},Z_j$ for $j\in[v]$. 

\textbf{Part 5:} suppose for a contradiction that $t \in \GDO{\mathbf{S^x}} \cap \GUP{\mathbf{S^z}}$. Take $i,j$ so that $S^z_j \leq t \leq S^x_i$. Take $k \in X_i \cap Z_j$. Then, $n_k-r^z_k \leq t_k \leq r^z_k -1$. Then, $r^x_k+r^z_k \geq n_k+1$, which contradicts $H^x_k (H^z_k)^T=0$. 
\end{proof}

\textbf{Running example:} For the running example of this section, the component decompositions \eqref{eq:ComponentDecompositionX} and \eqref{eq:ComponentDecompositionZ} are
\begin{align}
H^x_i&=\onebytwo{1}{1}=1\onebytwo{1}{0}\twobytwo{1}{1}{0}{1} \\
H^z_i&=\onebytwo{1}{1}=1\onebytwo{0}{1}\twobytwo{1}{0}{1}{1}
\end{align}
The matrices $\mathbf{S^x}$ and $\mathbf{S^z}$ from equations \eqref{eq:DecreasingSetGenerators} and \eqref{eq:IncreasingSetGenerators} are
\begin{align}
\mathbf{S^x}=\twobyfour{1}{1}{1}{1}{1}{1}{1}{1}-\twobyfour{1}{1}{0}{0}{0}{0}{1}{1}&=\twobyfour{0}{0}{1}{1}{1}{1}{0}{0} \\
\mathbf{S^z}&=\twobyfour{1}{0}{1}{0}{0}{1}{0}{1}
\end{align}
The incomplete permutation matrix
\begin{equation}
D(\mathbf{H^x},X_0)=\onebytwo{1}{0}\otimes\onebytwo{1}{0}\otimes\twobytwo{1}{0}{0}{1}\otimes\twobytwo{1}{0}{0}{1}
\end{equation}
is supported on columns indexed by
\begin{equation}
\GDO{\{S^x_0\}}=\{0000,0001,0010,0011\}
\end{equation}
where shorthand notation is used for elements of $\{0,1\}^4$. The incomplete permutation matrix
\begin{equation}
D(\mathbf{H^x},X_1)=\twobytwo{1}{0}{0}{1}\otimes\twobytwo{1}{0}{0}{1}\otimes\onebytwo{1}{0}\otimes\onebytwo{1}{0}
\end{equation}
is supported on columns indexed by
\begin{equation}
\GDO{\{S^x_1\}}=\{0000,0100,1000,1100\}
\end{equation}
The incomplete permutation matrix $D(\mathbf{H^x},\mathbf{X})$ is supported on columns indexed by
\begin{equation}
\GDO{\mathbf{S^x}}=\GDO{\{S^x_0\}} \cup \GDO{\{S^x_1\}} = \{0000,0001,0010,0011,0100,1000,1100\}
\end{equation}
Similarly, the incomplete permutation matrix
\begin{equation}
D(\mathbf{H^z},Z_0)=\onebytwo{0}{1}\otimes\twobytwo{1}{0}{0}{1}\otimes\onebytwo{0}{1}\otimes\twobytwo{1}{0}{0}{1}
\end{equation}
is supported on columns indexed by
\begin{equation}
\GUP{\{S^z_0\}}=\{1010,1011,1110,1111\}
\end{equation}
The incomplete permutation matrix
\begin{equation}
D(\mathbf{H^z},Z_1)=\twobytwo{1}{0}{0}{1}\otimes\onebytwo{0}{1}\otimes\twobytwo{1}{0}{0}{1}\otimes\onebytwo{0}{1}
\end{equation}
is supported on columns indexed by
\begin{equation}
\GUP{\{S^z_1\}}=\{0101,0111,1101,1111\}
\end{equation}
The incomplete permutation matrix $D(\mathbf{H^z},\mathbf{Z})$ is supported on columns indexed by
\begin{equation}
\GUP{\mathbf{S^z}}=\GUP{\{S^z_0\}} \cup \GUP{\{S^z_1\}} = \{0101,0111,1010,1011,1101,1110,1111\}
\end{equation}

\subsection{Block size and rate}\label{subsec:BlockSizeAndRate}

\begin{theorem}
Let $\mathbf{S^x},\mathbf{S^z}$ be as in Theorem \ref{thm:ColumnSupports}, equations \eqref{eq:DecreasingSetGenerators}, \eqref{eq:IncreasingSetGenerators}, and let $\mathbf{K}=[\vec{n}]\backslash(\GDO{\mathbf{S^x}}\cup\GUP{\mathbf{S^z}})$. Then, $CSS(\mathbf{H^x},\mathbf{H^z},\mathbf{X},\mathbf{Z})$ is a $[[\prod_{i=0}^{m-1} n_i, |\mathbf{K}|]]$ quantum CSS code. 
\end{theorem}

\begin{proof}
The matrix decompositions \eqref{eq:DecompositionOfMHxX} and \eqref{eq:DecompositionOfMHzZ} and the computation of the column supports of $D(\mathbf{H^x},\mathbf{X})$ and $D(\mathbf{H^z},\mathbf{Z})$ in Theorem \ref{thm:ColumnSupports} imply that $$rank(M(\mathbf{H^x},\mathbf{X}))=|\GDO{\mathbf{S^x}}| \text{ and } rank(M(\mathbf{H^z},\mathbf{Z}))=|\GUP{\mathbf{S^z}}|$$ Moreover, $\GDO{\mathbf{S^x}}$ and $\GUP{\mathbf{S^z}}$ are disjoint, again by Theorem \ref{thm:ColumnSupports}. Then, \newline $CSS(\mathbf{H^x},\mathbf{H^z},\mathbf{X},\mathbf{Z})$ has block size $\prod_{i=0}^{m-1} n_i$ qubits and has 
\begin{equation}
\prod_{i=0}^{m-1} n_i - rank(M(\mathbf{H^x},\mathbf{X})) - rank(M(\mathbf{H^z},\mathbf{Z})) = |\mathbf{K}|
\end{equation}
encoded qubits. 
\end{proof}

\textbf{Running example:} for the running example of this section, 
\begin{equation}
\mathbf{K}=\{0,1\}^4 \backslash (\GDO{\mathbf{S^x}}\cup\GUP{\mathbf{S^z}}) = \{0110,1001\}
\end{equation}

\subsection{Basis of the stabilizer}\label{subsec:BasisForStabilizer}

\begin{theorem}\label{thm:BasisForStabilizer}
The non-zero rows of $D(\mathbf{H^x},\mathbf{X})R(\mathbf{H^x})Q(\mathbf{H^x})^{-1}$ are a basis for $RowSpan(M(\mathbf{H^x},\mathbf{X}))$ and the non-zero rows of $D(\mathbf{H^z},\mathbf{Z})R(\mathbf{H^z})Q(\mathbf{H^z})^{-1}$ are a basis for $RowSpan(M(\mathbf{H^z},\mathbf{Z}))$.
\end{theorem}

\begin{proof}
Recall the matrix decomposition \eqref{eq:DecompositionOfMHxX}. The non-zero rows of \newline$D(\mathbf{H^x},\mathbf{X})R(\mathbf{H^x})Q(\mathbf{H^x})^{-1}$ are linearly independent because  $R(\mathbf{H^x}),Q(\mathbf{H^x})$ are invertible and $D(\mathbf{H^x},\mathbf{X})$ is an incomplete permutation matrix. Moreover, the row spans of $D(\mathbf{H^x},\mathbf{X})R(\mathbf{H^x})Q(\mathbf{H^x})^{-1}$ and $M(\mathbf{H^x},\mathbf{X})$ are the same, because $P(\mathbf{H^x},\mathbf{X}),L(\mathbf{H^x},\mathbf{X})$ are invertible. Similarly, \eqref{eq:DecompositionOfMHzZ}, implies the statement for the $Z$ stabilizers. 
\end{proof}

\textbf{Running example:} for the running example of this section, a basis of $RowSpan(M(\mathbf{H^x},\mathbf{X}))$ is obtained from the rows of $\twobytwo{1}{1}{0}{1}^{\otimes 4}$ indexed by $\GDO{\mathbf{S^x}}$, and a basis of $RowSpan(M(\mathbf{H^z},\mathbf{Z}))$ is obtained from the rows of $\twobytwo{1}{0}{1}{1}^{\otimes 4}$ indexed by $\GUP{\mathbf{S^z}}$. These are shown in figures \ref{fig:2dproductindependentstabilizersX} and \ref{fig:2dproductindependentstabilizersZ}. Note that the basis of the stabilizer contains two elements of weight 16, eight elements of weight 8, and four elements of weight 4. On the other hand, the dependent generators of the stabilizer given by $M(\mathbf{H^x},\mathbf{X}),M(\mathbf{H^z},\mathbf{Z})$ contain 16 elements of weight 4. This illustrates one advantage in using the matrices $M(\mathbf{H^x},\mathbf{X}),M(\mathbf{H^z},\mathbf{Z})$ to specify syndrome measurements. 

\subsection{Encoding circuit}\label{subsec:EncodingCircuit}
In describing the encoding circuit, it is convenient to index qubits by tuples in $[\vec{n}]$. 
\begin{theorem}\label{thm:EncodingCircuit}
Let $\mathbf{S^x},\mathbf{S^z}$ be as in Theorem \ref{thm:ColumnSupports}, equations \eqref{eq:DecreasingSetGenerators}, \eqref{eq:IncreasingSetGenerators}, and let $\mathbf{K}=[\vec{n}]\backslash(\GDO{\mathbf{S^x}}\cup\GUP{\mathbf{S^z}})$. Then, $CSS(\mathbf{H^x},\mathbf{H^z},\mathbf{X},\mathbf{Z})$ has the encoding circuit: 
\begin{enumerate}
\item Prepare qubits indexed by $\GDO{\mathbf{S^x}}$ in state $|+\rangle$. 
\item Prepare qubits indexed by $\GUP{\mathbf{S^z}}$ in state $|0\rangle$.
\item Prepare qubits indexed by $\mathbf{K}$ in the state to be encoded. 
\item Apply CNOT circuit $U \in \CliffordGroup{[\vec{n}]}$, where the image of $U$ under the surjective group homomorphism $\CliffordGroupHomomorphism$ satisfies 
\begin{equation}\label{eq:CNOTcircuitViaSymplecticImage}
\CliffordGroupHomomorphism(U^{-1})=\twobytwo{R(\mathbf{H^x})Q(\mathbf{H^x})^{-1}}{0}{0}{R(\mathbf{H^z})Q(\mathbf{H^z})^{-1}}
\end{equation}
\end{enumerate}
\end{theorem}

\begin{proof}
Before the application of $U$, the state to be encoded is placed in the joint +1 eigenspace of the single-qubit Pauli operators ${\PauliX}_i, i \in \GDO{\mathbf{S^x}}$ and ${\PauliZ}_j, j \in \GUP{\mathbf{S^z}}$. After the application of $U$, the encoded state is in the joint $+1$ eigenspace of the Pauli operators $U{\PauliX}_iU^{-1}, i \in \GDO{\mathbf{S^x}}$ and $U {\PauliZ}_j U^{-1}, j \in \GUP{\mathbf{S^z}}$. Using \eqref{eq:PropertyOfCliffordGroupHomomorphism} and \eqref{eq:CNOTcircuitViaSymplecticImage}, deduce $\PauliGroupHomomorphism(U{\PauliX}_iU^{-1})$ is row $i$ of $\onebytwo{R(\mathbf{H^x})Q(\mathbf{H^x})^{-1}}{0}$ and $\PauliGroupHomomorphism(U{\PauliZ}_jU^{-1})$ is row $j$ of $\onebytwo{0}{R(\mathbf{H^z})Q(\mathbf{H^z})^{-1}}$. Then, Theorem \ref{thm:BasisForStabilizer} implies that the final stabilizer subspace can equivalently be described using the parity check matrix $CSS(\mathbf{H^x},\mathbf{H^z},\mathbf{X},\mathbf{Z})$. 
\end{proof}

Moreover, the encoding circuit can be parallelized in the following sense:

\begin{theorem}
Suppose that there exist $d_0,\dots,d_{m-1}$ such that for each $i\in[m]$, the symplectic matrix $$\twobytwo{R_iQ_i^{-1}}{0}{0}{R_i^{-T}Q_i^{-1}}$$ has a preimage in the Clifford group given by a depth $d_i$ CNOT circuit. Then, the symplectic matrix $$\twobytwo{R(\mathbf{H^x})Q(\mathbf{H^x})^{-1}}{0}{0}{R(\mathbf{H^z})Q(\mathbf{H^z})^{-1}}$$ has a preimage in the Clifford group given by a depth $\sum_{i=0}^{m-1} d_i$ CNOT circuit. 
\end{theorem}

\begin{proof}
Write
\begin{multline}\label{eq:EncodingCircuitFactorization}\twobytwo{R(\mathbf{H^x})Q(\mathbf{H^x})^{-1}}{0}{0}{R(\mathbf{H^z})Q(\mathbf{H^z})^{-1}} \\ = \prod_{i=0}^{m-1}\twobytwo{I \otimes \dots \otimes R_iQ_i^{-1} \otimes \dots \otimes I}{0}{0}{I \otimes \dots \otimes R_i^{-T}Q_i^{-1} \otimes \dots \otimes I}\end{multline}
The $i$-th term has preimage in the Clifford group that is a depth $d_i$ CNOT circuit: the preimage of $$\twobytwo{R_iQ_i^{-1}}{0}{0}{R_i^{-T}Q_i^{-1}}$$ performed in parallel on all groups of qubits $PROJ^{-1}_{\{i\}^c}(y), \;\; y \in \left(\times_{j \neq i} [n_j]\right)$, where $PROJ_{\{i\}^c} : [\vec{n}] \rightarrow \left(\times_{i \in X_1^c}[n_i]\right)$ is the coordinate projection on positions other than $i$. 
\end{proof}

\textbf{Running example:} For the running example of this section, 
$$\twobytwo{R(\mathbf{H^x})Q(\mathbf{H^x})^{-1}}{}{}{R(\mathbf{H^z})Q(\mathbf{H^z})^{-1}} = \twobytwo{\twobytwo{1}{1}{0}{1}^{\otimes 4}}{}{}{\twobytwo{1}{0}{1}{1}^{\otimes 4}}$$

The symplectic matrix
$$ \begin{pmatrix} 1&1&0&0\\0&1&0&0\\0&0&1&0\\0&0&1&1 \end{pmatrix} $$
has preimage in the Clifford group a CNOT gate with control qubit 0 and target qubit 1. The symplectic matrix
$$\twobytwo{I_8 \otimes \twobytwo{1}{1}{0}{1}}{}{}{I_8 \otimes \twobytwo{1}{0}{1}{1}}$$ has a depth 1 preimage in the Clifford group: CNOT gates with control qubit $(y,0)$ and target qubit $(y,1)$ for $y \in \{0,1\}^3$. Similarly, the other three terms in the encoding circuit factorization \eqref{eq:EncodingCircuitFactorization} have preimages of depth 1. The full depth 4 encoding circuit can be seen in Figure \ref{fig:2dproductencodingcircuit}. The CNOT part of the circuit follows the same recursive pattern as is used to encode Polar and Reed-Muller codes. 

\begin{figure}[H]
\centering
\caption{Independent generators of the stabilizer and encoding circuit for the running example.}
\label{fig:2dproductindependentstabilizers}
\begin{subfigure}{\textwidth}
\centering
\caption{Non-zero elements of the rows of $\twobytwo{1}{1}{0}{1}^{\otimes 4}$ indexed by $\GDO{\mathbf{S^x}}$}
\label{fig:2dproductindependentstabilizersX}
\begin{tikzpicture}[scale=0.5]
\foreach \x in {0,...,16} {\draw (\x,0) -- (\x,-7);}
\foreach \y in {0,...,7} \draw (0,-\y) -- (16, -\y) ;
\foreach \p/\q/\r/\s/\y in {0/0/0/0/0,0/0/0/1/1,0/0/1/0/2,0/0/1/1/3,0/1/0/0/4,1/0/0/0/5,1/1/0/0/6}
{
\foreach \a in {0,1}
	\foreach \b in {0,1}
		\foreach \c in {0,1}
			\foreach \d in {0,1}
				{
					\tikzmath{\x=8*\a+4*\b+2*\c+\d;}
					\tikzmath{\z=(1-\p+\p*\a)*(1-\q+\q*\b)*(1-\r+\r*\c)*(1-\s+\s*\d);}
					\fill (\x,-\y)--(\x+\z,-\y)--(\x+\z,-\y-\z)--(\x,-\y-\z)--(\x,-\y);					
				}
}
\end{tikzpicture}
\end{subfigure}
\begin{subfigure}{\textwidth}
\centering
\caption{Non-zero elements of the rows of $\twobytwo{1}{0}{1}{1}^{\otimes 4}$ indexed by $\GUP{\mathbf{S^z}}$}
\label{fig:2dproductindependentstabilizersZ}
\begin{tikzpicture}[scale=0.5]
\foreach \x in {0,...,16} {\draw (\x,0) -- (\x,-7);}
\foreach \y in {0,...,7} \draw (0,-\y) -- (16, -\y) ;
\foreach \p/\q/\r/\s/\y in {0/1/0/1/0,0/1/1/1/1,1/0/1/0/2,1/0/1/1/3,1/1/0/1/4,1/1/1/0/5,1/1/1/1/6}
{
\foreach \a in {0,1}
	\foreach \b in {0,1}
		\foreach \c in {0,1}
			\foreach \d in {0,1}
				{
					\tikzmath{\x=8*\a+4*\b+2*\c+\d;}
					\tikzmath{\z=(\p+(1-\p)*(1-\a))*(\q+(1-\q)*(1-\b))*(\r+(1-\r)*(1-\c))*(\s+(1-\s)*(1-\d));}
					\fill (\x,-\y)--(\x+\z,-\y)--(\x+\z,-\y-\z)--(\x,-\y-\z)--(\x,-\y);					
				}
}
\end{tikzpicture}
\end{subfigure}
\begin{subfigure}{\textwidth}
\centering
\caption{Encoding circuit}
\label{fig:2dproductencodingcircuit}
\begin{tikzpicture}[scale=0.5]
{%Qubit index, circuit lines
\foreach \a in {0,1}
\foreach \b in {0,1}
\foreach \c in {0,1}
\foreach \d in {0,1}
{
\tikzmath{\y=8*\a+4*\b+2*\c+\d;}
\draw (-4,\y) node[fill=white] {$Q_{\a\b\c\d}$};
\draw (0,\y)--(13.5,\y);
}
}
{% input states |+\rangle
\foreach \a/\b/\c/\d in {0/0/0/0,0/0/0/1,0/0/1/0,0/0/1/1,0/1/0/0,1/0/0/0,1/1/0/0}
{
\tikzmath{\y=8*\a+4*\b+2*\c+\d;}
\draw (-1,\y) node {$|+\rangle$};
}
}
{%\input states |0\rangle
\foreach\a/\b/\c/\d in {0/1/0/1,0/1/1/1,1/0/1/0,1/0/1/1,1/1/0/1,1/1/1/0,1/1/1/1}
{
\tikzmath{\y=8*\a+4*\b+2*\c+\d;}
\draw (-1,\y) node {$|0\rangle$};
}
}
{%CNOT gates
\foreach \p in {0,1}
\foreach \q in {0,1}
\foreach \r in {0,1}
{
\tikzmath{\xd=1;}
\tikzmath{\ydc=8*\p+4*\q+2*\r;}
\tikzmath{\ydt=8*\p+4*\q+2*\r+1;}
\draw (\xd,\ydc)--(\xd,\ydt+0.2);
\filldraw (\xd,\ydc) circle [radius=0.2];
\draw (\xd,\ydt) circle [radius=0.2];
\tikzmath{\xc=3+0.5*\r;}
\tikzmath{\ycc=8*\p+4*\q+\r;}
\tikzmath{\yct=8*\p+4*\q+2+\r;}
\draw (\xc,\ycc)--(\xc,\yct+0.2);
\filldraw (\xc,\ycc) circle [radius=0.2];
\draw (\xc,\yct) circle [radius=0.2];
\tikzmath{\xb=5.5+\q+0.5*\r;}
\tikzmath{\ybc=8*\p+2*\q+\r;}
\tikzmath{\ybt=8*\p+4+2*\q+\r;}
\draw (\xb,\ybc)--(\xb,\ybt+0.2);
\filldraw (\xb,\ybc) circle [radius=0.2];
\draw (\xb,\ybt) circle [radius=0.2];
\tikzmath{\xa=9+2*\p+\q+0.5*\r;}
\tikzmath{\yac=4*\p+2*\q+\r;}
\tikzmath{\yat=8+4*\p+2*\q+\r;}
\draw (\xa,\yac)--(\xa,\yat+0.2);
\filldraw (\xa,\yac) circle [radius=0.2];
\draw (\xa,\yat) circle [radius=0.2];
}
}
\end{tikzpicture}
\end{subfigure}
\end{figure}

\subsection{Normalizer, logical operators}\label{subsec:NormalizerLogicalOperators}

\begin{theorem}
The normalizer for $CSS(\mathbf{H^x},\mathbf{H^z},\mathbf{X},\mathbf{Z})$ is generated by $iI$ and the Pauli group elements corresponding to the non-zero rows of
\begin{equation}\label{eq:NormalizerGenerators}
\twobytwo{\GP{\GUP{\mathbf{S^z}}^c}R(\mathbf{H^x})Q(\mathbf{H^x})^{-1}}{}{}{\GP{\GDO{\mathbf{S^x}}^c}R(\mathbf{H^z})Q(\mathbf{H^z})^{-1}}
\end{equation}

A canonical choice of logical operators for $CSS(H^x,H^z,X,Z)$ is obtained from the non-zero rows of
\begin{equation}\label{eq:LogicalOperators}
\twobytwo{\GP{\mathbf{K}}R(\mathbf{H^x})Q(\mathbf{H^x})^{-1}}{}{}{\GP{\mathbf{K}}R(\mathbf{H^z})Q(\mathbf{H^z})^{-1}}
\end{equation}
\end{theorem}

\begin{proof}
Think of the encoding circuit in Theorem \ref{thm:EncodingCircuit}. Before the application of $U$, the subspace stabilized by the single qubit Pauli operators ${\PauliX}_i,i \in \GDO{\mathbf{S^x}}$ and ${\PauliZ}_j,j \in \GUP{\mathbf{S^z}}$ has a canonical choice for generators of the normalizer: $iI$ and the single qubit Pauli operators ${\PauliX}_i,i\in\GUP{\mathbf{S^z}}^c$ and ${\PauliZ}_j,j\in\GDO{\mathbf{S^x}}^c$. It also has a canonical choice of logical operators: the single qubit Pauli operators ${\PauliX}_i,{\PauliZ}_i, i \in \mathbf{K}$. After the application of $U$, the single qubit Pauli operators in the original set of normalizer generators and logical operators are conjugated by $U$ to become the Pauli operators corresponding to the rows of $\CliffordGroupHomomorphism(U^{-1})$ given in \eqref{eq:NormalizerGenerators} and \eqref{eq:LogicalOperators}. 
\end{proof}

\textbf{Running example:} For the running example, the $X$ logical operators are determined by the rows of $\twobytwo{1}{1}{0}{1}^{\otimes 4}$ indexed by $\mathbf{K}=\{0110,1001\}$. The first one of these corresponds to Pauli-$\PauliX$ on qubits $Q_{0000},Q_{0001},Q_{1001},Q_{1000}$, and the second to Pauli-$\PauliX$ on qubits $Q_{0000},Q_{0010},Q_{0110},Q_{0100}$. Similarly, the $Z$ logical operators are given by the rows of $\twobytwo{1}{0}{1}{1}^{\otimes 4}$ indexed by $\mathbf{K}=\{0110,1001\}$. The first of these corresponds to Pauli-$\PauliZ$ on qubits $Q_{0000},Q_{0010},Q_{0110},Q_{0100}$ and the second to Pauli-$\PauliZ$ on qubits $Q_{0000},Q_{0001},Q_{1001},Q_{1000}$. The logical operators can be visualized using Figure \ref{fig:2dproducttannergraph}. It can be seen there that they correspond to the non-contractible loops of the torus. 

\subsection{Row and column weight}\label{subsec:RowAndColumnWeight}

While the present paper is focused on quantum CSS codes with block size a few tens or a few hundreds of qubits and with sparse parity check matrices, it is worth mentioning the asymptotic behaviour of the row and column weights of the parity check matrices $CSS(\mathbf{H^x},\mathbf{H^z},\mathbf{X},\mathbf{Z})$. 

\begin{theorem}\label{thm:RowAndColumnWeight}
Let $w_r,w_c$ be upper bounds on the weight of rows, respectively columns, of matrices in the tuples $\mathbf{H^x},\mathbf{H^z}$. Let $p$ be an upper bound on the size of the sets in the tuples $\mathbf{X},\mathbf{Z}$. Then, $w_r^p$ is an upper bound on the weight of rows of $CSS(\mathbf{H^x},\mathbf{H^z},\mathbf{X},\mathbf{Z})$ and on the degree of checks in the Tanner graph, $\max(u,v)w_c^p$ is an upper bound on the weight of columns of $CSS(\mathbf{H^x},\mathbf{H^z},\mathbf{X},\mathbf{Z})$ and $(u+v)w_c^p$ is an upper bound on the degree of qubits in the Tanner graph.  
\end{theorem}

\begin{proof}
Each layer $$M(\mathbf{H^x},X_i)=\otimes_{j=0}^{m-1} \begin{cases}H_j & \text{if } j \in X_i \\ I_{n_j}& \text{otherwise}\end{cases}$$ has rows of weight at most $w_r^{|X_i|}$ and columns of weight at most $w_c^{|X_i|}$. The same applies to layers of $M(\mathbf{H^z},\mathbf{Z})$. Maximizing over the layers gives the upper bound on row weight/check degree, and summing over the layers gives the upper bound on column weight/qubit degree. 
\end{proof}

Therefore, keeping $u,v,p,w_r,w_c$ constant and letting (some of) the $n_j$ go to infinity gives quantum ldpc codes. The properties of asymptotic families obtained in this way are left for future research, while the rest of this paper focuses on the finite block length regime. 

\textbf{Running example:} for the running example of this section, $w_r=p=u=v=2$, $w_c=1$. The parity check matrix $CSS(\mathbf{H^x},\mathbf{H^z},\mathbf{X},\mathbf{Z})$ has row weight 4 and column weight 2. In the Tanner graph, all vertices have degree 4. 

\subsection{Syndrome measurement circuit}\label{subsec:SyndromeMeasurementCircuit}

\begin{theorem}
Let $p=\max\left( \left\{ \prod_{j \in X_i^c} n_j : i \in [u] \right\} \cup \left\{ \prod_{j \in Z_i^c} n_j : i \in [v] \right\} \right)$. Let $d,a$ be such that the following two statements hold:
\begin{enumerate}
\item For each $i \in [u]$, the $\PauliX$ measurements specified by $\otimes_{j \in X_i} H^x_j$ can be measured by a circuit with at most $a$ ancilla qubits and depth at most $d$. 
\item For each $i \in [v]$, the $\PauliZ$ measurements specified by $\otimes_{j \in Z_i} H^z_j$ can be measured by a circuit with at most $a$ ancilla qubits and depth at most $d$. 
\end{enumerate}
Then, $CSS(\mathbf{H^x},\mathbf{H^z},\mathbf{X},\mathbf{Z})$ has a syndrome measurement circuit with at most $pa$ ancilla qubits and depth at most $(u+v)d$. 
\end{theorem}

\begin{proof}
Consider the layer $M(\mathbf{H^x},X_0)$. There exist row and column permutations that map $M(\mathbf{H^x},X_0)$ to
\begin{equation}
\left(\otimes_{j\in X_0^c} I_{n_j}\right) \otimes \left( \otimes_{j \in X_0} H^x_j \right) = \threebythree{\otimes_{j \in X_0} H^x_j}{}{}{}{\ddots}{}{}{}{\otimes_{j \in X_0} H^x_j}
\end{equation}
Therefore, the syndrome measurements specified by $M(\mathbf{H^x},X_0)$ can be performed with at most $pa$ ancilla qubits and depth at most $d$: perform the circuit for $\otimes_{i \in X_0} H^x_i$ in parallel on all groups of qubits $PROJ_{X_0^c}^{-1}(y), \;\; y \in \left(\times_{i \in X_0^c}[n_i]\right) $ where $PROJ_{X_0^c} : [\vec{n}] \rightarrow \left(\times_{i \in X_0^c}[n_i]\right)$ is the coordinate projection on the positions in $X_0^c$.  A similar argument applies for the other layers. Then, all syndrome measurements of $CSS(\mathbf{H^x},\mathbf{H^z},\mathbf{X},\mathbf{Z})$ can be performed by a circuit with at most $pa$ ancilla qubits and depth at most $(u+v)d$. 
\end{proof}

\textbf{Running example:} For the running example of this section, $p=4$. The row vector $\onebytwo{1}{1}\otimes\onebytwo{1}{1}=\onebyfour{1}{1}{1}{1}$ specifies a syndrome measurement that can be performed with one ancilla qubit and depth 6: preparation of the ancilla in state $|+\rangle$ or state $|0\rangle$, 4 CNOT gates and measurement of the ancilla in the $\PauliX$ or $\PauliZ$ eigenbasis. Therefore, the running example $CSS(\mathbf{H^x},\mathbf{H^z},\mathbf{X},\mathbf{Z})$ has a syndrome measurement circuit with 4 ancilla qubits and depth 24. 

\subsection{Protection against syndrome errors}\label{subsec:ProtectionAgainstSyndromeErrors}

The vector spaces $Im(M(\mathbf{H^x},\mathbf{X}))$ and $Im(M(\mathbf{H^z},\mathbf{Z}))$ can be called the spaces of valid syndromes: if no error occurs during the measurements, then the syndromes for phase flips and bit flips belong respectively to $Im(M(\mathbf{H^x},\mathbf{X}))$ and $Im(M(\mathbf{H^z},\mathbf{Z}))$. When there are errors in the measurement outcomes, then $Im(M(\mathbf{H^x},\mathbf{X}))$, $Im(M(\mathbf{H^z},\mathbf{Z}))$ can be viewed as classical linear codes. The matrix decompositions \eqref{eq:DecompositionOfMHxX} and \eqref{eq:DecompositionOfMHzZ} give the following information about these vector spaces:

\begin{theorem}
\begin{enumerate}
\item $Im(M(\mathbf{H^x},\mathbf{X}))$ is a \newline $\left[\sum_{j=0}^{u-1} \left(\prod_{i\in X_j} l_i^x\right)\left(\prod_{i \in X_j^c} n_i \right), \left|\GDO{\mathbf{S^x}}\right|\right]$ classical linear code. A basis for it is given by the non-zero columns of $P(\mathbf{H^x},\mathbf{X})^{-1} L(\mathbf{H^x},\mathbf{X})D(\mathbf{H^x},\mathbf{X})$. A basis for its orthogonal complement is given by rows of \newline $L(\mathbf{H^x},\mathbf{X})^{-1}P(\mathbf{H^x},\mathbf{X})$ indexed by the complement of the row support of $D(\mathbf{H^x},\mathbf{X})$. 
\item $Im(M(\mathbf{H^z},\mathbf{Z}))$ is a  $\left[\sum_{j=0}^{v-1} \left(\prod_{i\in Z_j} l_i^z\right)\left(\prod_{i \in Z_j^c} n_i \right), \left|\GUP{\mathbf{S^z}}\right|\right]$ classical linear code. A basis for it is given by the non-zero columns of \newline $P(\mathbf{H^z},\mathbf{Z})^{-1} L(\mathbf{H^z},\mathbf{Z})D(\mathbf{H^z},\mathbf{Z})$. A basis for its orthogonal complement is given by rows of $L(\mathbf{H^z},\mathbf{Z})^{-1}P(\mathbf{H^z},\mathbf{Z})$ indexed by the complement of the row support of $D(\mathbf{H^z},\mathbf{Z})$. 
\end{enumerate}
\end{theorem}

\begin{proof}
A single layer $M(\mathbf{H^x},X_j)$ has $\left(\prod_{i\in X_j} l^x_i\right)\left(\prod_{i \in X_j^c} n_i\right)$ rows. Therefore, the total number of rows of $M(\mathbf{H^x},\mathbf{X})$ is $\sum_{j=0}^{u-1} \left(\prod_{i\in X_j} l_i^x\right)\left(\prod_{i \in X_j^c} n_i \right)$. The dimension of $Im(M(\mathbf{H^x},\mathbf{X}))$ is $rank(M(\mathbf{H^x},\mathbf{X}))=\left| \GDO{\mathbf{S^x}} \right|$. Now, use \eqref{eq:DecompositionOfMHxX} to find a basis for $Im(M(\mathbf{H^x},\mathbf{X}))$ and its orthogonal complement. Since $Q(\mathbf{H^x})$ and $R(\mathbf{H^x})$ are invertible, $Im(M(\mathbf{H^x},\mathbf{X}))=Im(P(\mathbf{H^x},\mathbf{X})^{-1}L(\mathbf{H^x},\mathbf{X})D(\mathbf{H^x},\mathbf{X}))$. Since $P(\mathbf{H^x},\mathbf{X})$ and $L(\mathbf{H^x},\mathbf{X})$ are invertible and $D(\mathbf{H^x},\mathbf{X})$ is an incomplete permutation matrix, a basis for  this vector space is given by the non-zero columns of $P(\mathbf{H^x},\mathbf{X})^{-1}L(\mathbf{H^x},\mathbf{X})D(\mathbf{H^x},\mathbf{X})$, which are also the columns of \newline $P(\mathbf{H^x},\mathbf{X})^{-1}L(\mathbf{H^x},\mathbf{X})$ indexed by the row support of $D(\mathbf{H^x},\mathbf{X})$. Therefore, a basis for the orthogonal complement is given by rows of $L(\mathbf{H^x},\mathbf{X})^{-1}P(\mathbf{H^x},\mathbf{X})$ indexed by the complement of the row support of $D(\mathbf{H^x},\mathbf{X})$. This proves part one. Similarly, \eqref{eq:DecompositionOfMHzZ} implies part two. 
\end{proof}

\textbf{Running example:} The matrix decomposition of $M(\mathbf{H^x},\mathbf{X})$ for the running example of this section is 
\begin{multline}
M(\mathbf{H^x},\mathbf{X})=L(\mathbf{H^x},\mathbf{X})D(\mathbf{H^x},\mathbf{X})R(\mathbf{H^x}) \\
=\left(\twobytwo{\twobytwo{1}{1}{0}{1}^{\otimes 2}}{}{}{\twobytwo{1}{1}{0}{1}^{\otimes 2}} \Big(I_8 + e_{8,4}e_{8,0}^T \Big)\right) \\ * \left(\Big(I_8 + e_{8,4}e_{8,0}^T \Big) \twobyone{D(\mathbf{H^x},X_0)}{D(\mathbf{H^x},X_1)} \right)\twobytwo{1}{1}{0}{1}^{\otimes 4}
\end{multline}
where $e_{n,0}, \dots, e_{n,n-1}$ denotes the standard basis of $\F^n$. The incomplete permutation matrix $D(\mathbf{H^x},\mathbf{X})=\Big(I_8 + e_{8,4}e_{8,0}^T \Big) \twobyone{D(\mathbf{H^x},X_0)}{D(\mathbf{H^x},X_1)}$  is supported on rows $0,1,2,3,5,6,7$. Therefore, a basis for $Im(M(\mathbf{H^x},\mathbf{X}))$ is given by columns $0,1,2,3,5,6,7$ of 
\begin{multline}
L(\mathbf{H^x},\mathbf{X})=\twobytwo{\twobytwo{1}{1}{0}{1}^{\otimes 2}}{}{}{\twobytwo{1}{1}{0}{1}^{\otimes 2}} \Big(I_8 + e_{8,4}e_{8,0}^T \Big) \\ =
\begin{pmatrix}
1&1&1&1&0&0&0&0\\
0&1&0&1&0&0&0&0\\
0&0&1&1&0&0&0&0\\
0&0&0&1&0&0&0&0\\
1&0&0&0&1&1&1&1\\
0&0&0&0&0&1&0&1\\
0&0&0&0&0&0&1&1\\
0&0&0&0&0&0&0&1
\end{pmatrix}
\end{multline}
and a basis for the orthogonal complement of $Im(M(\mathbf{H^x},\mathbf{X}))$ is given by row 4 of
\begin{multline}
L(\mathbf{H^x},\mathbf{X})^{-1}=\Big(I_8 + e_{8,4}e_{8,0}^T \Big) \twobytwo{\twobytwo{1}{1}{0}{1}^{\otimes 2}}{}{}{\twobytwo{1}{1}{0}{1}^{\otimes 2}}\\ =
\begin{pmatrix}
1&1&1&1&0&0&0&0\\
0&1&0&1&0&0&0&0\\
0&0&1&1&0&0&0&0\\
0&0&0&1&0&0&0&0\\
1&1&1&1&1&1&1&1\\
0&0&0&0&0&1&0&1\\
0&0&0&0&0&0&1&1\\
0&0&0&0&0&0&0&1
\end{pmatrix}
\end{multline}
Therefore, $Im(M(\mathbf{H^x},\mathbf{X}))$ is the $[8,7,2]$ single parity check code. Similarly, $Im(M(\mathbf{H^z},\mathbf{Z}))$ is also the $[8,7,2]$ single parity check code.  

\section{A subfamily related to a generalization of Reed-Muller codes}\label{sec:SpecialCase}

The simplest class of examples from the construction is obtained when all matrices in the tuples $\mathbf{H^x},\mathbf{H^z}$ are $\onebytwo{1}{1}$. The matrices $M(\mathbf{H},\mathbf{X})$ from equation \eqref{eq:LayeredTensorProductMatrixNotation} in section \ref{sec:MatricesWithLayersOfTensorProducts} in this case depend only on $\mathbf{X}$, so the shorthand notation $M(\mathbf{X}),CSS(\mathbf{X},\mathbf{Z})$ will be used in this section. The equations \eqref{eq:DecreasingSetGenerators}, \eqref{eq:IncreasingSetGenerators} determining the generators $\mathbf{S^x},\mathbf{S^z}$ of the decreasing, respectively increasing, subset of column indices simplify to 
\begin{align}
\mathbf{S^x} &= \ones{u} \ones{m}^T - \mathbf{X} \\
\mathbf{S^z} &= \mathbf{Z}
\end{align}

Properties of examples in this class follow from properties of a family of vector spaces that generalize classical Reed-Muller codes. 

\subsection{Generalization of Reed-Muller codes}\label{sec:GRMcodes}

Let
\begin{equation}
R_m=\twobytwo{1}{1}{0}{1}^{\otimes m}
\end{equation}

To decreasing $\mathbf{S} \subset\{0,1\}^m$ associate the vector space 
\begin{equation}\label{eq:VectorSpacesGeneralizingReedMullerCodes}
\GRM{\mathbf{S}}=RowSpan \left( \GP{\mathbf{S}} R_m\right)
\end{equation}
and to increasing $\mathbf{T} \subset\{0,1\}^m$ associate the vector space
\begin{equation}
\GRMT{\mathbf{T}}=RowSpan \left( \GP{\mathbf{T}} R_m^{-T} \right)
\end{equation}
It will be seen later that these are two different parametrizations of the same family of vector spaces. 

The vector spaces $\GRM{\mathbf{S}}$ for $\mathbf{S}$ decreasing generalize Reed-Muller codes. If
$
\mathbf{S}(r,m)=\{s\in\{0,1\}^m:|s|\leq r\}
$
then the corresponding subspace is 
$
\GRM{\mathbf{S}(r,m)}=\mathcal{RM}(r,m),
$
the usual Reed-Muller code with parameters $r,m$. This is because the row of $R_m$ with index $s\subset [m]$ is the truth table of the monomial $\mu_s(v_0,\dots,v_{m-1})=\prod_{i \in s} v_i$, so $RowSpan\left(\GP{\mathbf{S}(r,m)}R_m\right)$ is the span of the truth tables of all monomials of degree at most $r$.

Moreover, it turns out that some of the properties of Reed-Muller codes can be extended to the vector spaces $\GRM{\mathbf{S}}$. 

First, the following theorem shows that $\GRM{\mathbf{S}},\GRMT{\mathbf{T}}$ are two \newline parametrizations of the same family of vector spaces. It also generalizes \cite[Chapter 13, Theorems 4 and 12]{macwilliams1977theory}: it shows that $\GRM{\mathbf{S}}$ is spanned by certain low-weight elements and that the orthogonal complement of such a vector space is another vector space of the same family. 

\begin{theorem}\label{thm:LinearAlgebraPropertiesOfGRM}
Take a tuple of subsets of $[m]$
\begin{equation}
\mathbf{S}=\threebyone{S_0}{\vdots}{S_{u-1}}
\end{equation} 
and take
$
\mathbf{T}=\ones{u}\ones{m}^T - \mathbf{S}
$.
Then,
\begin{equation}\label{eq:TwoParametrizationsOfGRM}
\GRM{\GDO{\mathbf{S}}}=\GRMT{\GUP{\mathbf{T}}}=RowSpan(M(\mathbf{T}))
\end{equation}
and
\begin{equation}\label{eq:OrthogonalComplementOfGRM}
\GRM{\GDO{\mathbf{S}}}^\bot = \GRM{\GUP{\mathbf{T}}^c}
\end{equation}
\end{theorem}

\begin{proof}
For the pair of orthogonal matrices $\onebytwo{1}{1},\onebytwo{1}{1}$, Theorem \ref{thm:JointDecompositionForPairOfOrthogonalMatrices} gives the decompositions
\begin{align}
\onebytwo{1}{1}&=\onebytwo{1}{0}\twobytwo{1}{1}{0}{1} \label{eq:OneOneDecompositionOne}\\
\onebytwo{1}{1}&=\onebytwo{0}{1}\twobytwo{1}{0}{1}{1}\label{eq:OneOneDecompositionTwo}
\end{align}

From the two decompositions \eqref{eq:OneOneDecompositionOne} and \eqref{eq:OneOneDecompositionTwo}, and from equation \eqref{eq:LDUforLayersOfTensorProducts} in section \ref{sec:MatricesWithLayersOfTensorProducts} follow two alternative decompositions of $M(\mathbf{T})$:
\begin{equation}\label{eq:MXdecompositions}
M(\mathbf{T})=L(\mathbf{T}) \GP{\GDO{\mathbf{S}}} R_m= L'(\mathbf{T}) \GP{\GUP{\mathbf{T}}} R_m^{-T}
\end{equation}
Equation \eqref{eq:MXdecompositions} implies the Theorem. Indeed, since $L(\mathbf{T}),L'(\mathbf{T})$ are invertible, \eqref{eq:MXdecompositions} imples \eqref{eq:TwoParametrizationsOfGRM}. Moreover, the orthogonal complement of \newline $RowSpan\left(\GP{\GUP{\mathbf{T}}}R_m^{-T}\right)$ is $RowSpan\left(\GP{\GUP{\mathbf{T}}^c}R_m\right) $, and this gives \eqref{eq:OrthogonalComplementOfGRM}.
\end{proof}

Moreover, the inductive argument using the $(u,u+v)$ construction for the distance of Reed-Muller codes \cite[Chapter 13, Theorems 2 and 3]{macwilliams1977theory} also generalizes to the vector spaces $\GRM{\mathbf{S}}$, and to the distances of quantum CSS codes built from these vector spaces. For $\mathbf{S},\mathbf{T}$ decreasing, $\mathbf{S} \subset \mathbf{T}$, let 
\begin{align}
r(\mathbf{T},\mathbf{S})&=\max \{|t|:t \in \mathbf{T}, t\notin \mathbf{S}\}\\
d(\mathbf{T},\mathbf{S})&=\min \{|v|:v\in\GRM{\mathbf{T}}, v \notin \GRM{\mathbf{S}}\}
\end{align}
Then: 
\begin{theorem}\label{thm:DistanceForNestedGRM}
Let $\mathbf{S},\mathbf{T} \subset \{0,1\}^m$ be decreasing with $\mathbf{S} \subset \mathbf{T}$. Then, 
$$
d(\mathbf{T},\mathbf{S})=2^{m-r(\mathbf{T},\mathbf{S})}
$$
\end{theorem}

\begin{proof}
Take
\begin{align}
\mathbf{S}_0&=\{s\in\{0,1\}^{m-1}:(0,s)\in \mathbf{S}\}\\
\mathbf{S}_1&=\{s\in\{0,1\}^{m-1}:(1,s)\in \mathbf{S}\}\\
\mathbf{T}_0&=\{t\in\{0,1\}^{m-1}:(0,t)\in \mathbf{T}\}\\
\mathbf{T}_1&=\{t\in\{0,1\}^{m-1}:(1,t)\in \mathbf{T}\}
\end{align}
so that
\begin{align}
\mathbf{S}&=(0,\mathbf{S}_0)\cup(1,\mathbf{S}_1) \\
\mathbf{T}&=(0,\mathbf{T}_0)\cup(1,\mathbf{T}_1)
\end{align}

Note that
$
r(\mathbf{T},\mathbf{S})=\max(r(\mathbf{T}_0,\mathbf{S}_0),r(\mathbf{T}_1,\mathbf{S}_1)+1)
$. 

Note further that
\begin{align}
\GP{\mathbf{S}}R_m&=\twobytwo{\GP{\mathbf{S}_0}}{}{}{\GP{\mathbf{S}_1}}\twobytwo{R_{m-1}}{R_{m-1}}{}{R_{m-1}}\\
\GP{\mathbf{T}}R_m&=\twobytwo{\GP{\mathbf{T}_0}}{}{}{\GP{\mathbf{T}_1}}\twobytwo{R_{m-1}}{R_{m-1}}{}{R_{m-1}}
\end{align}
so $\GRM{\mathbf{S}}$ is obtained from $\GRM{\mathbf{S}_0}$ and $\GRM{\mathbf{S}_1}$ using the $(u,u+v)$ construction and, similarly, $\GRM{\mathbf{T}}$ is obtained from $\GRM{\mathbf{T}_0}$ and $\GRM{\mathbf{T}_1}$ using the $(u,u+v)$ construction. In the special case $\mathbf{S}=\emptyset, \GRM{\mathbf{S}}=\{0\}$, the formula for the distance of the $(u,u+v)$ construction \cite[Chapter 2, Theorem 33]{macwilliams1977theory} is sufficient to finish the argument. In the general case, Theorem \ref{thm:NestedGRMUUV} in the present paper gives 
$
d(\mathbf{T},\mathbf{S})=\min(2d(\mathbf{T}_0,\mathbf{S}_0),d(\mathbf{T}_1,\mathbf{S}_1))
$. 
The proof is completed by induction on $m$. 
\end{proof}

The remaining step of the proof of Theorem \ref{thm:DistanceForNestedGRM} is

\begin{theorem}\label{thm:NestedGRMUUV}
For $\mathbf{S},\mathbf{T},\mathbf{S}_0,\mathbf{S}_1,\mathbf{T}_0,\mathbf{T}_1$ as in Theorem \ref{thm:DistanceForNestedGRM}, 
$$
d(\mathbf{T},\mathbf{S})=\min(2d(\mathbf{T}_0,\mathbf{S}_0),d(\mathbf{T}_1,\mathbf{S}_1))
$$
\end{theorem}

\begin{proof}
First, $\GRM{\mathbf{T}}\backslash\GRM{\mathbf{S}}$ contains vectors of the form $(u,u)$ for $u \in \GRM{\mathbf{T}_0}\backslash\GRM{\mathbf{S}_0}$ and $(0,v)$ for $v\in \GRM{\mathbf{T}_1}\backslash\GRM{\mathbf{S}_1}$. Therefore, 
$
d(\mathbf{T},\mathbf{S})\leq\min(2d(\mathbf{T}_0,\mathbf{S}_0),d(\mathbf{T}_1,\mathbf{S}_1))
$.

It remains to prove the other direction. Note that, since $\mathbf{S},\mathbf{T}$ are decreasing, $\mathbf{S}_0,\mathbf{S}_1,\mathbf{T}_0,\mathbf{T}_1$ are also decreasing and $\mathbf{S}_1 \subset \mathbf{S}_0$, $\mathbf{T}_1 \subset \mathbf{T}_0$. Moreover, $\mathbf{S} \subset \mathbf{T}$ implies $\mathbf{S}_0 \subset \mathbf{T}_0$, $\mathbf{S}_1 \subset \mathbf{T}_1$. 

Now, take any $w \in\GRM{\mathbf{T}}\backslash\GRM{\mathbf{S}}$ and write $w=(u,u+v)$ for $u \in\GRM{\mathbf{T}_0}, v\in\GRM{\mathbf{T}_1}$. Take $t \in \mathbf{T} \backslash \mathbf{S}$ such that, when $w$ is expressed as a linear combination of the rows of $\GP{\mathbf{T}}R_m$, it has a non-zero coefficient for row $t$. Write $t=(b,t'),b\in\{0,1\},t'\in\{0,1\}^{m-1}$ and consider cases:

\paragraph{Case 1: $b=1,t'\in \mathbf{T}_1 \backslash \mathbf{S}_1$.} Then, $v \in \GRM{\mathbf{T}_1}\backslash \GRM{\mathbf{S}_1}$, so
\begin{equation}
|(u,u+v)|=|u|+|u+v|\geq |v| \geq d(\mathbf{T}_1,\mathbf{S}_1)
\end{equation}

\paragraph{Case 2: $b=0,t'\in \mathbf{T}_0 \backslash (\mathbf{S}_0 \cup \mathbf{T}_1)$.} Then, $$(u,u+v) \in \GRM{\mathbf{T}_0}\backslash\GRM{\mathbf{S}_0\cup \mathbf{T}_1}$$ so
\begin{equation}
|(u,u+v)|=|u|+|u+v| \geq 2d(\mathbf{T}_0,\mathbf{S}_0 \cup \mathbf{T}_1) \geq 2d(\mathbf{T}_0,\mathbf{S}_0)
\end{equation}
where the last step follows from $\mathbf{S}_0 \subset \mathbf{S}_0 \cup \mathbf{T}_1$. 

\paragraph{Case 3: $b=0,t'\in \mathbf{T}_1 \backslash \mathbf{S}_0$.} Then, $u \in \GRM{\mathbf{T}_1}\backslash\GRM{\mathbf{T}_1 \cap \mathbf{S}_0}$, so
\begin{equation}
|(u,u+v)| \geq |u| \geq d(\mathbf{T}_1,\mathbf{T}_1 \cap \mathbf{S}_0) \geq d(\mathbf{T}_1, \mathbf{S}_1)
\end{equation}
where the last step follows from $\mathbf{S}_1 \subset \mathbf{T}_1 \cap \mathbf{S}_0$. 
\end{proof}

It will be convenient later to have an expression for the minimum distance also in terms of the other parametrization of the family of vector spaces. For $\mathbf{S},\mathbf{T} \subset \{0,1\}^m$ increasing, $\mathbf{S} \subset \mathbf{T}$, let 
\begin{align}
r'(\mathbf{T},\mathbf{S}) &= \max\{m-|t|:t \in \mathbf{T}\backslash \mathbf{S}\} \\
d'(\mathbf{T},\mathbf{S}) &= \min\{|v|: v \in \GRMT{\mathbf{T}} \backslash \GRMT{\mathbf{S}}\}
\end{align}
Then:
\begin{corollary}\label{cor:DistanceForNestedGRMT}
Let $\mathbf{S},\mathbf{T} \subset \{0,1\}^m$ be decreasing, $\mathbf{S} \subset \mathbf{T}$. Then, $
d'(\mathbf{T},\mathbf{S})=2^{m-r'(\mathbf{T},\mathbf{S})}
$. 
\end{corollary}

\begin{proof}
Let $\phi$ be the function that switches ones with zeros and zeros with ones. 

Theorem \ref{thm:LinearAlgebraPropertiesOfGRM} implies
$
\GRM{\phi(\mathbf{S})}=\GRMT{\mathbf{S}}$ and \newline $
\GRM{\phi(\mathbf{T})}=\GRMT{\mathbf{T}}
$.
Then, $d'(\mathbf{T},\mathbf{S})=d(\phi(\mathbf{T}),\phi(\mathbf{S}))$. 

Moreover, the maximum number of zeros in an element of $\mathbf{T}\backslash \mathbf{S}$ is equal to the maximum number of ones in an element of $\phi(\mathbf{T})\backslash\phi(\mathbf{S})$, so $r'(\mathbf{T},\mathbf{S})=r(\phi(\mathbf{T}),\phi(\mathbf{S}))$. Theorem \ref{thm:DistanceForNestedGRM} completes the proof. 
\end{proof}

\subsection{Properties of the $CSS(\mathbf{X},\mathbf{Z})$ codes}

The $CSS(\mathbf{X},\mathbf{Z})$ codes inherit from the general construction all properties in section \ref{sec:GeneralCase}. Note in particular that these codes have especially simple encoding and syndrome measurement circuits, as illustrated already by the running example in Section \ref{sec:GeneralCase}. In addition, the connection to generalized Reed-Muller codes and the results of section \ref{sec:GRMcodes} give formulas for the distances. 

\begin{theorem}\label{thm:DistancesOfCSScodeFromGRM}
For $\mathbf{X},\mathbf{Z} \subset \{0,1\}^m$ satisfying \eqref{eq:FullyIntersectingCondition}, let 
$
\mathbf{S^x} = \ones{u}\ones{m}^T - \mathbf{X}$, $\mathbf{S^z} = \mathbf{Z}$, 
$\mathbf{K} = \{0,1\}^m \backslash\left( \GDO{\mathbf{S^x}} \cup \GUP{\mathbf{S^z}} \right)$. Then, the distances of $CSS(\mathbf{X},\mathbf{Z})$ are 
\begin{align}
d_x&=2^{\min\{m-|v|:v\in \mathbf{K}\}} \\
d_z&=2^{\min\{|v|:v \in \mathbf{K}\}}
\end{align}
\end{theorem}

\begin{proof}
From the results of section \ref{sec:GeneralCase} deduce that 
$d_x = d(\GUP{\mathbf{S^z}}^c,\GDO{\mathbf{S^x}})$ and $d_z = d'(\GDO{\mathbf{S^x}}^c,\GUP{\mathbf{S^z}})$. Then, from Theorem \ref{thm:DistanceForNestedGRM} and Corollary \ref{cor:DistanceForNestedGRMT} deduce that $d_x = 2^{m-r(\GUP{\mathbf{S^z}}^c,\GDO{\mathbf{S^x}})}$ and $d_z = 2^{m-r'(\GDO{\mathbf{S^x}}^c,\GUP{\mathbf{S^z}})}$. 
Finally, note that $r(\GUP{\mathbf{S^z}}^c,\GDO{\mathbf{S^x}})=\max\{|v|:v \in \mathbf{K}\}$ and $ r'(\GDO{\mathbf{S^x}}^c,\GUP{\mathbf{S^z}})=\max\{m-|v|:v \in \mathbf{K}\}$
to complete the proof. 
\end{proof}

Moreover, the connection to generalized Reed-Muller codes allows the computation of the minimum distance of $Im(M(\mathbf{S}))$ for any tuple of subsets $\mathbf{S}$. Recall that in the context of quantum error correction, the vector space $Im(M(\mathbf{S}))$ has the following interpretation: if the outcome of the syndrome measurements specified by $M(\mathbf{S})$ is error-free, then it is a vector in $Im(M(\mathbf{S}))$. Therefore, the minimum distance of $Im(M(\mathbf{S}))$ is a measure of the ability to correct errors in the syndrome. 

\begin{theorem}\label{thm:MinDistOfImMS}
Take any tuple $\mathbf{S}$ consisting of $u$ subsets $S_0,\dots, S_{u-1}$ of $[m]$. The minimum weight of a non-zero element of $Im(M(\mathbf{S}))$ is
\begin{equation}\label{eq:MinDistOfImMS}
\min \left\{ 2^{|T|}|\{i:S_i \cap T=\emptyset \}| : T \subset [m], \exists i :T \cap S_i = \emptyset \right\}
\end{equation}
\end{theorem}

\begin{proof}
First, $Im(M(\mathbf{S}))=Im(M(\mathbf{S})R_m)$, because $R_m$ is invertible. 

Next, for each $i \in [u]$, 
\begin{equation}
M(S_i)R_m = \otimes_{j=0}^{m-1} \begin{cases} \onebytwo{1}{0} & \text{if } j \in S_i \\ \twobytwo{1}{1}{0}{1} & \text{otherwise} \end{cases}
\end{equation}
Therefore, $M(S_i)R_m$ has the $2^{m-|S_i|}$ columns of $R_{m-|S_i|}$ placed in positions indexed by $T \subset [m]$ such that $T \cap S_i = \emptyset$ and all the remaining columns are zero. 

From the above, deduce that the weight of the column of $M(\mathbf{S})R_m$ indexed by $T \subset[m]$ is $2^{|T|} |\{i : S_i \cap T = \emptyset\}|$. Then, the expression \eqref{eq:MinDistOfImMS} gives the minimum weight of a non-zero column of $M(\mathbf{S})R_m$, and, therefore, it is an upper bound on the minimum distance of $Im(M(\mathbf{S}))$. It remains to prove that it is also a lower bound. 

Take any collection of $v$ subsets $T_0,\dots,T_{v-1}$ and consider the sum \newline $\sum_j M(\mathbf{S})R_m e_{T_j}$ of the corresponding columns of $M(\mathbf{S})R_m$ (here $e_{T_j}$ denotes the standard basis vector corresponding to $T_j$). Suppose without loss of generality that each column $M(\mathbf{S})R_m e_{T_j}$ in the sum is non-zero. 

Partition the sets $S_0, \dots, S_{v-1}$ in two groups depnding on whether \newline$\sum_j M(S_i)R_m e_{T_j}$ is zero or non-zero. Without loss of generality, $\sum_j M(S_i)R_m e_{T_j}$ is non-zero for $i=0, \dots, k-1$ and it is zero for $i=k, \dots u-1$. 

Let $T$ be a subset of $[m]$ of minimum size subject to the constraint that it intersects each of $S_{k}, \dots, S_{u-1}$ and that there is some $i \in [k]$ such that $T$ does not intersect $S_i$. 

\textbf{Claim 1:} each $T_j$ satisfies the given constraint, so $|T| \leq |T_j|$ for each $j \in [v]$. 

\textbf{Proof of Claim 1:} Note that for each $i$, the non-zero columns of $M(S_i)R_m$ are linearly independent. Then, $\sum_j M(S_i)R_m e_{T_j} =0$ implies $\forall j: S_i \cap T_j \neq \emptyset$. Deduce that for all $i=k, \dots, u-1$, for all $j \in [v]$, $S_i \cap T_j \neq \emptyset$. Finally, since each column $M(\mathbf{S})R_me_{T_j}$ is non-zero, deduce that $\forall j \in [v], \exists i \in [k]: T_j\cap S_i=\emptyset$. 

\textbf{Claim 2:} $\left| \sum_j M(\mathbf{S})R_me_{T_j} \right| \geq 2^{|T|}|\{i:S_i \cap T = \emptyset\}$. 

\textbf{Proof of Claim 2:} Take any $i \in [k]$. $\sum_j M(S_i)R_me_{T_j}$ is a non-zero linear combination of certain columns of $R_{m-|S_i|}$. Corollary \ref{cor:DistanceForNestedGRMT} implies that its weight is at least $2^{\min\{|T_j|:T_j\cap S_i=\emptyset\}}$. Claim 1 implies that this is at least $2^{|T|}$. Summing over $i$ proves Claim 2. 

Claim 2 implies that the expression \eqref{eq:MinDistOfImMS} is a lower bound on the minimum distance of $Im(M(\mathbf{S}))$, and completes the proof of the theorem. 
\end{proof}

\section{Examples}\label{sec:Examples}

This section is divided in three parts. 

First, \ref{sec:StandardRM} considers quantum CSS codes based on standard Reed-Muller codes. The quantum distances of these are half the number of qubits in a syndrome measurement. Except for the few examples with block size $\leq16$ these CSS codes have syndrome measurements on more than 8 qubits. 

The other two subsections consider quantum CSS codes based on generalized Reed-Muller codes. These provide much greater flexibility in designing quantum LDPC codes, including codes for which the distance is greater than the syndrome measurement weight. 

\ref{sec:SymmetricExamples} gives a number of examples for which $d_x=d_z$. These examples are suitable for noise models that are symmetric with respect to bit flip and phase flip errors, for example the depolarizing channel or the quantum erasure channel. 

Finally, \ref{sec:AsymmetricExamples} gives a number of examples with $d_x \neq d_z$. These examples are better suited for biased noise, where either bit flip errors occur with higher probability than phase flip errors or vice versa. 

The examples in this section are in the subfamily from section \ref{sec:SpecialCase}; all components of the matrix tuples $\mathbf{H^x},\mathbf{H^z}$ are $H^x_i=H^z_i=\onebytwo{1}{1}$, and the shorthand notation $M(\mathbf{X}),M(\mathbf{Z}),CSS(\mathbf{X},\mathbf{Z})$ is used. 

\subsection{Quantum CSS codes based on standard Reed-Muller codes}\label{sec:StandardRM}

Note that Theorem \ref{thm:DistancesOfCSScodeFromGRM} implies that for CSS codes obtained from generalized Reed-Muller codes, 
\begin{equation}\label{eq:DistanceUpperBound}
d_xd_z \leq 2^m,
\end{equation} 
with equality if and only if all elements of $\mathbf{K}$ have the same weight.

Now take\footnote{In the rest of the paper, $\mathbf{X},\mathbf{Z}$ are "tuples of subsets" or "tuples of indicator vectors". However, sometimes, it is more convenient to give $\mathbf{X},\mathbf{Z}$ as "sets of subsets" or "sets of indicator vectors". To convert the latter to the former, the lexicographic order (or any other order) can be used. \label{footnote:SetsAndTuples}}
\begin{align}
\mathbf{X} &= \{ v \in \{0,1\}^m:|v|=m-r+1\} \\
\mathbf{Z} &= \{v \in \{0,1\}^m:|v|=r+1\} \\
\GDO{\mathbf{S^x}}&=\{v \in \{0,1\}^m: |v|\leq r-1\} \\
\mathbf{K} &= \{v \in \{0,1\}^m: |v|=r\} \\
\GUP{\mathbf{S^z}} &= \{v \in \{0,1\}^m:|v| \geq r+1\}
\end{align}
Then, $CSS(\mathbf{X},\mathbf{Z})$ is a $[[2^m,{m \choose r}]]$ code with
\begin{enumerate}
\item ${m \choose r-1}2^{r-1}$ product-of-Pauli-$\PauliX$ syndrome measurements on $2^{m-r+1}$ qubits each.
\item ${m \choose r+1}2^{m-r-1}$ product-of-Pauli-$\PauliZ$ syndrome measurements on $2^{r+1}$ qubits each.
\item Distances $d_x=2^{m-r}$ and $d_z=2^r$. 
\end{enumerate}
These codes have been considered previously in in \cite{steane1999quantumreedmuller} and in \cite[Lemma 4.1]{sarvepalli2009asymmetric}.  
 
$CSS(\mathbf{X},\mathbf{Z})$ achieves equality $d_xd_z=2^m$ in \eqref{eq:DistanceUpperBound}, and, moreover, it does so optimally because it uses the entire layer $\{v \in \{0,1\}^m: |v|=r\}$ for logical operators. However, $CSS(\mathbf{X},\mathbf{Z})$ also has an undesirable property: the weight of syndrome measurements is twice the respective distance. 

If it is desired that syndrome measurements involve at most 8 qubits, then only examples with block size $\leq 16$ remain. The most interesting of these is the $[[16,6,4]]$ quantum code where both $Ker(M(\mathbf{X}))$ and $Ker(M(\mathbf{Z}))$ are the Reed-Muller(2,4) code. This example matches the best possible distance for a $[[16,6]]$ stabilizer code\footnote{Obtained from \cite{grassl2007codetables}, specifically \url{http://codetables.de/QECC.php?q=4&n=16&k=6}}. The rows of the matrices $M(\mathbf{X}),M(\mathbf{Z})$ specify 16 measurements of 8 qubits each. The spaces of valid syndromes $Im(M(\mathbf{X}))$, $Im(M(\mathbf{Z}))$ are [8,5,2] classical linear codes. 

Below, it will be seen how sacrificing part of the layer $\{v \in \{0,1\}^m: |v|=r\}$ for syndrome measurements leads to codes with better trade-off between distance and syndrome measurement weight.

\subsection{Examples based on generalized Reed-Muller codes with $d_x=d_z$}\label{sec:SymmetricExamples}

\subsubsection{Quantum codes based on classical product codes}

\paragraph{CSS code from the 2D product of SPC codes}

This is a CSS code where each of $M(\mathbf{X}),M(\mathbf{Z})$ is isomorphic (by row and column permutation) to the parity check matrix of the 2D product of the $[4,3]$ single parity check code. For more details, see the running example in Section \ref{sec:GeneralCase}.

\paragraph{CSS code from the 3D product of SPC codes}

This is a CSS code where each of $M(\mathbf{X}),M(\mathbf{Z})$ is isomorphic (by row and column permutation) to the parity check matrix of the 3D product of the $[8,7]$ single parity check code. The code appears previously in \cite{ostrev2024classicalproduct}. The code has parameters $[[512,174,8]]$, has 384 syndrome measurements of weight 8 and is obtained by choosing $m=9$ and 
\begin{align}
\mathbf{X}&=\threebyone{\{0,1,2\}}{\{3,4,5\}}{\{6,7,8\}} &\mathbf{Z}&=\threebyone{\{0,3,6\}}{\{1,4,7\}}{\{2,5,8\}}
\end{align}
i.e. $\mathbf{X},\mathbf{Z}$ correspond to the rows and columns of
\begin{equation}\label{eq:ThreeByThreeSquare}
\threebythree{0}{1}{2}{3}{4}{5}{6}{7}{8}
\end{equation}
The spaces of valid syndromes $Im(M(\mathbf{X}))$, $Im(M(\mathbf{Z}))$ are [192,169,3] classical linear codes. 

\subsubsection{Examples based on a cyclic pattern}

In these examples, $X,Z$ are obtained from several cyclic shifts of the first subset. 

A $[[32,14,4]]$ code with 24 syndrome measurements of weight 8 is obtained by choosing $m=5$ and
\begin{equation}
\mathbf{X}=\mathbf{Z}=\threebyone{\{0,1,3\}}{\{1,2,4\}}{\{2,3,0\}}
\end{equation}
The spaces of valid syndromes $Im(M(\mathbf{X}))$, $Im(M(\mathbf{Z}))$ are [12,9,2] classical linear codes. 

A $[[64,8,8]]$ code with 96 syndrome measurements of weight 8 is obtained by choosing $m=6$ and
\begin{align}
\mathbf{X}=\mathbf{Z}&=\{013,124,235,340,451,502\}\\
\mathbf{K}&=\{012,123,234,345,450,501,024,135\}
\end{align}
where shorthand notation omitting curly brackets and commas is used for subsets of $\{0,1,2,3,4,5\}$ (see also footnote \ref{footnote:SetsAndTuples}). The spaces of valid syndromes $Im(M(\mathbf{X}))$, $Im(M(\mathbf{Z}))$ are [48,28,4] classical linear codes.  

A $[[128,10,8]]$ with 192 syndrome measurements of weight 8 is obtained by choosing $m=7$ and
\begin{align}
\mathbf{X}=\mathbf{Z}&=\{013,124,235,346,450,561\} \\
\mathbf{K}&=\{345,145,135,134,1345,026,0256,0246,0236,0126\}
\end{align}
The spaces of valid syndromes $Im(M(\mathbf{X}))$, $Im(M(\mathbf{Z}))$ are [96,59,4] classical linear codes. 

\subsubsection{A further example with block size 128}

A $[[128,24,8]]$ code with 160 syndrome measurements of weight 8 is obtained by choosing $m=7$ and
\begin{align}
\mathbf{X}&=\{012,013,234,356,456\}\\
\mathbf{Z}&=\{143,146,360,325,025\}
\end{align}
The spaces of valid syndromes $Im(M(\mathbf{X}))$, $Im(M(\mathbf{Z}))$ are [80,52,4] classical linear codes. 

\subsubsection{Examples with distance greater than the syndrome measurement weight}

A $[[256,6,16]]$ code with 512 syndrome measurements of weight 8 is obtained by choosing $m=8$ and
\begin{align}
\mathbf{X}&=\{012,123,234,345,456,567,670,701\}\\
\mathbf{Z}&=\{136,247,350,461,572,603,714,025\}\\
\mathbf{K}&=\{2367,1357,1256,0347,0246,0145\}
\end{align}
The spaces of valid syndromes $Im(M(\mathbf{X}))$, $Im(M(\mathbf{Z}))$ are [256,125,8] classical linear codes. 

A $[[512,18,16]]$ code with 768 syndrome measurements of weigth 8 is obtained by taking $m=9$,
\begin{align}
\mathbf{X}&=\{012,345,678,048,156,237\}\\
\mathbf{Z}&=\{036,147,258,246,138,057\}
\end{align}
i.e. $\mathbf{X}$ corresponds to the rows and one set of diagonals, and $\mathbf{Z}$ to the columns and the other set of diagonals of \eqref{eq:ThreeByThreeSquare}. 
The spaces of valid syndromes $Im(M(\mathbf{X}))$, $Im(M(\mathbf{Z}))$ are [384,247,6] classical linear codes.

\subsection{Examples based on generalized Reed-Muller codes with $d_x \neq d_z$}\label{sec:AsymmetricExamples}

Note that in any of the examples of these section, the role of $\mathbf{X},\mathbf{Z}$ can be switched, depending on whether bit flip or phase flip errors are more likely in the noise model of interest. 

\subsubsection{Highly asymmetric CSS codes with distances $2,2^{m-1}$, $m=3,4,\dots$}

Take
\begin{align}
\mathbf{X}&=\{\{0\}\} \\
\mathbf{Z}&=\{\{0,i\}:i=1,\dots,m-1\}\} \\
\GDO{\mathbf{S^x}}&=\{v \in \{0,1\}^m: v \leq 01\dots 1\} \\
\mathbf{K} &= \{10\dots 0\} \\
\GUP{\mathbf{S^z}} &= \{v \in \{0,1\}^m:\exists i, v \geq \{0,i\}\}
\end{align}
Then, $CSS(\mathbf{X},\mathbf{Z})$ is a $[[2^m,1]]$ code with
\begin{enumerate}
\item $2^{m-1}$ product-of-Pauli-$\sigma_x$ syndrome measurements on $2$ qubits each.
\item $(m-1)2^{m-2}$ product-of-Pauli-$\sigma_z$ syndrome measurements on $4$ qubits each.
\item Distances $d_x=2^{m-1}$ and $d_z=2$. 
\end{enumerate}
The spaces of valid syndromes $Im(M(\mathbf{X}))$, $Im(M(\mathbf{Z}))$ are $[2^{m-1},2^{m-1},1]$ and $[(m-1)2^{m-2},2^{m-1}-1,m-1]$ classical linear codes.

\subsubsection{Block size 32}

A $[[32,2]]$ code with distances $d_x=8, d_z=4$, $48$ syndrome measurements of weight 4 and 4 syndrome measurements of weight 8 is obtained by choosing:
\begin{align}
\mathbf{X} &=\{\{0,1\},\{2,3,4\}\}\\
\mathbf{Z} &=\{\{0,2\},\{1,3\},\{0,4\},\{1,4\},\{1,3\}\}\\
\mathbf{K} &=\{\{0,3\},\{1,2\}\} 
\end{align}
The spaces of valid syndromes $Im(M(\mathbf{X}))$, $Im(M(\mathbf{Z}))$ are [12,11,2] and [40,19,4] classical linear codes.  

Another $[[32,2]]$ code with distances $d_x=8, d_z=4$ is obtained by choosing:
\begin{align}
\mathbf{X} &=\{\{0,1,4\},\{2,3,4\}\}\\
\mathbf{Z} &=\{\{0,2\},\{1,3\},\{4\}\}\\
\mathbf{K} &=\{\{0,3\},\{1,2\}\} 
\end{align}
This code has 8 syndrome measurements of weight 8, 16 syndrome measurements of weight 4, and 16 syndrome measurements of weight 2. The spaces of valid syndromes $Im(M(\mathbf{X}))$, $Im(M(\mathbf{Z}))$ are [8,7,2] and [32,23,3] classical linear codes. 

It turns out that this code has the same stabilizer as a toric code.\footnote{The author would like to thank the anonymous reviewer who pointed out that a refinement of the tesselation of the torus used in Figure \ref{fig:2dproducttannergraph} leads to a $[[32,2,d_x=8,d_z=4]]$ toric code.} This can be seen by a sequence of row operations that transform the layers $M(Z_0),M(Z_1)$ of the matrix $M(\mathbf{Z})$. The row operations are conveniently described using the shorthand notation $h=\onebytwo{1}{1}, e_0^T=\onebytwo{1}{0},e_1^T=\onebytwo{0}{1}, I=e_0e_0^T+e_1e_1^T$. Now, note that 
\begin{equation}
M(Z_0)+\left(h\otimes I \otimes h \otimes I \otimes e_1\right) M(Z_2) = h \otimes I \otimes h \otimes I \otimes \twobytwo{1}{0}{1}{0}
\end{equation}
which is just two copies of $h \otimes I \otimes h \otimes I \otimes e_0^T$. Similarly, 
\begin{equation}
M(Z_1)+\left( I \otimes h \otimes I \otimes h\otimes e_0\right) M(Z_2) =  I \otimes h \otimes I \otimes h \otimes \twobytwo{0}{1}{0}{1}
\end{equation}
which is just two copies of $ I \otimes h \otimes I \otimes h \otimes e_1^T$. Thus, $M(\mathbf{Z})$ is equivalent by row operations to
\begin{equation}
M'(\mathbf{Z}) = \threebyone{M'(Z_0)}{M'(Z_1)}{M(Z_2)}=\threebyone{h \otimes I \otimes h \otimes I \otimes e_0^T}{I \otimes h \otimes I \otimes h \otimes e_1^T}{I \otimes I \otimes I \otimes I \otimes h}
\end{equation}
The two parity check matrices $M(\mathbf{X})$ and $M'(\mathbf{Z})$ specify a Tanner graph that can be embedded in the torus; this is shown in figure \ref{fig:n32k2torictannergraph}. 

\begin{figure}[H]
\centering
\caption{Tanner graph of the toric code equivalent to $CSS(\{014,234\},\{02,13,4\})$. For $a,b,c \in \{0,1\}$, $C^{X_a}_{bc}$ denotes the check associated to row $bc$ of $M(X_a)$, and $C^{Z_a}_{bc}$ denotes the check associated to row $bc$ of $M'(Z_a)$. For $a,b,c,d \in \{0,1\}$, $C^{Z_2}_{abcd}$ denotes the check associated to row $abcd$ of $M(Z_2)$. For $a,b,c,d,e \in \{0,1\}$, qubit $Q_{abcde}$ (not labelled in the picture) is the node between $C^{Z_2}_{abcd}$ and the nearest node of type $C^{Z_e}$. }
\label{fig:n32k2torictannergraph}
\begin{tikzpicture}
\foreach \z in {-2,2}
{
	\foreach \y in {-2,2}
	{
		\draw (-4,\y)--(-2,\y+\z)--(0,\y)--(2,\y+\z)--(4,\y);
	}
}
\foreach \x in {-4,0}
\foreach \y in {-4,0}
\foreach \z in {-0.5,0.5}
\foreach \w in {0,2}
{
\draw (\x+\w,\y+\w)--(\x+\w+1+\z,\y+\w+1-\z)--(\x+\w+2,\y+\w+2);
}
\foreach \x in {0,4}
\foreach \y in {-4,0}
\foreach \z in {-0.5,0.5}
\foreach \w in {0,2}
{
\draw (\x-\w,\y+\w)--(\x-\w-1+\z,\y+\w+1+\z)--(\x-\w-2,\y+\w+2);
}
\foreach \a in {0,1}
	\foreach \b in {0,1}
		{
			\tikzmath{\x=4-4*\a;}
			\tikzmath{\y=4-4*\b;}
			\draw (\x,\y) node[fill=white]{\tiny{$C_{\b\a}^{X_0}$}};
		}
\foreach \a in {0,1}
	\foreach \b in {0,1}
		{
			\tikzmath{\x=-2+4*\a;}
			\tikzmath{\y=-2+4*\b;}
			\draw (\x,\y) node[fill=white]{\tiny{$C_{\a\b}^{X_1}$}};
		}
\foreach \a in {0,1}
	\foreach \b in {0,1}
		{
			\tikzmath{\x=-2+4*\a;}
			\tikzmath{\y=4-4*\b;}
			\draw (\x,\y) node[fill=white]{\tiny{$C_{\a\b}^{Z_1}$}};
		}
\foreach \a in {0,1}
	\foreach \b in {0,1}
		{
			\tikzmath{\x=4-4*\a;}
			\tikzmath{\y=-2+4*\b;}
			\draw (\x,\y) node[fill=white]{\tiny{$C_{\b\a}^{Z_0}$}};
		}
\foreach \a in {0,1}
	\foreach \b in {0,1}
		\foreach \c in {0,1}
			\foreach \d in {0,1}
				{
					\tikzmath{\x=-3+6*\a+2*\d-4*\a*\d;}
					\tikzmath{\y=-3+6*\b+2*\c-4*\b*\c;}
					\draw (\x,\y) node[fill=white]{\tiny{$C^{Z_2}_{\a\b\c\d}$}};
				}
\end{tikzpicture}
\end{figure}

Note however that $M'(\mathbf{Z})$ has columns of weight 2, so the ability to correct one error in the bit flip syndrome is lost when $M'(\mathbf{Z})$ is used instead of $M(\mathbf{Z})$. 

\subsubsection{Block size 128}

A $[[128,3]]$ code with distances $d_x=8,d_z=16$ and with $224$ syndrome measurements on 8 qubits each is obtained by choosing
\begin{align}
\mathbf{X}&=\{013,124,235,346,450,561,602,134\} \\
\mathbf{Z}&=\{013,124,235,346,450,561\}\\
\mathbf{K}&=\{0246,0236,0126\}
\end{align}
The spaces of valid syndromes $Im(M(\mathbf{X}))$, $Im(M(\mathbf{Z}))$ are [128,66,8] and [96,59,4] classical linear codes.  

\section{Conclusion and future work}\label{sec:Conclusion}

The present paper introduced intersecting subset codes, established a number of useful properties, and gave many examples with small and moderate block sizes. 

One direction for future work concerns algorithms that prescribe a corretion based on the syndrome. Since intersecting subset codes have connections to classical LDPC, polar, and Reed-Muller codes, there are many possible low-complexity decoding algorithms that can potentially be applied. It would be interesting to evaluate the various options empirically and to determine which gives the best performance. 

A second direction for future work concerns the performance of intersecting subset codes when used on noisy near term quantum hardware. As already explained, for a large subfamily of intersecting subset codes it is possible to calculate the distances for both data and syndrome errors, and to construct a large number of interesting examples. A next step could be an investigation of the behavior of these examples under the standard error model for fault tolerance, in which each elementary gate in the syndrome measurement circuit is followed by independent Pauli errors.

A third direction for future work concerns further exploration of the structure of intersecting subset codes. As explained in the introduction, the code design was inspired by the relation between Gallager's LDPC codes and classical products of single parity check codes. There was no a priori reason to expect the connection to Reed-Muller codes in section \ref{sec:SpecialCase} or the two examples of intersecting subset codes that turn out to also be toric codes. These coincidences indicate that intersecting subset codes have rich and interesting structure that may hold further surprises.

\section*{Acknowledgements}

This work was supported by the QuantERA II Programme, through the project Error Correction for Quantum Information Processing (EQUIP). This project has received funding from the European Union's Horizon 2020 Research and Innovation Programme under Grant Agreement no. 731473 and 101017733. 

The author would like to thank Pradeep Sarvepalli and two anonynous reviewers for the detailed and constructive comments and useful suggestions for improvement.

\end{document}